\documentclass[journal]{IEEEtran}

\usepackage[utf8]{inputenc} % allow utf-8 input
\usepackage[T1]{fontenc}    % use 8-bit T1 fonts
\usepackage{hyperref}       % hyperlinks
\usepackage{url}            % simple URL typesetting
\usepackage{booktabs}       % professional-quality tables
\usepackage{amsfonts}       % blackboard math symbols
\usepackage{nicefrac}       % compact symbols for 1/2, etc.
\usepackage{microtype}      % microtypography
\usepackage{xcolor}         % colors
\usepackage{enumitem}   
\usepackage{amsmath}
\usepackage{amsthm}
\usepackage{amssymb}
\usepackage{adjustbox}
\usepackage{balance}
\usepackage{enumerate}
\usepackage{algorithm}
\usepackage{algpseudocode}
\usepackage{listings}
\usepackage{mathtools}
\usepackage{cite}

\newcommand{\beq}{\begin{equation}}
\newcommand{\eeq}{\end{equation}}
\def\bv{\mbox{\boldmath $v$}}

\def\by{\mbox{\boldmath $y$}}
\def\bv{\mbox{\boldmath $v$}}
\def\bw{\mbox{\boldmath $w$}}

\def\bx{\mathbf{x}}
\def\bg{\mbox{\boldmath $g$}}

\def\bt{\mbox{\boldmath $t$}}

\def\bz{\mbox{\boldmath $z$}}
\def\by{\mbox{\boldmath $y$}}

\def\bX{\mbox{\boldmath $X$}}

\def\bX{\mathbf{X}}

\def\b1{\mbox{\boldmath $1$}}
\def\b0{\mbox{\boldmath $0$}}

\def\mA{\mbox{$\mathbf{A}$}}
\def\mh{\mbox{$\mathbf{h}$}}

\def\mY{\mbox{$\mathbf{Y}$}}

\def\mP{\mbox{$\mathbf{P}$}}

\def\mZ{\mbox{$\mathbf{Z}$}}
\def\mC{\mbox{$\mathbf{C}$}}

\def\L{{\cal L}}

\newtheorem{definition}{Definition}
\newtheorem{lemma}{\textit{Lemma}}

\newtheorem{remark}{Remark}

\floatstyle{ruled}
\newfloat{algorithm}{tbp}{loa}
\providecommand{\algorithmname}{Algorithm}
\floatname{algorithm}{\protect\algorithmname}

\makeatletter
\newcommand*\dotp{\mathpalette\dotp@{.5}}
\newcommand*\dotp@[2]{\mathbin{\vcenter{\hbox{\scalebox{#2}{$\m@th#1\bullet$}}}}}
\makeatother

\def\D{{\mathbf D}}

\def\H{{\mathbf H}}

\def\h{{\mathbf h}}
\def\L{{\mathbf L}}

\usepackage{graphicx}
\usepackage{caption}
\usepackage{subcaption}
\usepackage{dsfont}

\newtheorem{Theorem}{\textbf{\textit{Theorem}}}

\begin{document}
\title{Topological Adaptive Least Mean Squares Algorithms over Simplicial Complexes\vspace{.1cm}}

\author{Lorenzo Marinucci, ~\IEEEmembership{Student Member,~IEEE,} Claudio Battiloro, ~\IEEEmembership{Member,~IEEE,} \\and Paolo Di~Lorenzo,~\IEEEmembership{Senior Member,~IEEE}
        \vspace{-.6cm}
        % <-this % stops a space
\thanks{Marinucci is with the Department of Statistics, Sapienza University of Rome, Piazzale Aldo Moro, 5, 00185 Rome, Italy.  Battiloro is a Postdoctoral Fellow at the Biostatistics Department of Harvard University, 655 Huntington Ave, Boston, MA 02115,
USA. Di Lorenzo is with the Department of Information Engineering, Electronics, and Telecommunications, Sapienza University of Rome, Via Eudossiana 18, 00184, Rome, Italy. {E-mail: \{l.marinucci, paolo.dilorenzo\}@uniroma1.it}, cbattiloro@hsph.harvard.edu. This work has been supported by the SNS JU project 6G-GOALS under the EU’s Horizon program Grant Agreement No 101139232, and by European Union under the Italian National Recovery and Resilience Plan of NextGenerationEU, partnership on Telecommunications of the Future (PE00000001- program RESTART). A preliminary version of this work was presented in \cite{marinucci2024topological}.}}

\maketitle
\begin{abstract}
This paper introduces a novel adaptive framework for processing dynamic flow signals over simplicial complexes, extending classical least-mean-squares (LMS) methods to high-order topological domains. Building on discrete Hodge theory, we present a topological LMS algorithm that efficiently processes streaming signals observed over time-varying edge subsets. We provide a detailed stochastic analysis of the algorithm, deriving its stability conditions, steady-state mean-square-error, and convergence speed, while exploring the impact of edge sampling on performance. We also propose strategies to design optimal edge sampling probabilities, minimizing rate while ensuring desired estimation accuracy. Assuming partial knowledge of the complex structure (e.g., the underlying graph), we introduce an adaptive topology inference method that integrates with the proposed LMS framework. Additionally, we propose a distributed version of the algorithm and analyze its stability and mean-square-error properties. Empirical results on synthetic and real-world traffic data demonstrate that our approach, in both centralized and distributed settings, outperforms graph-based LMS methods by leveraging higher-order topological features.
%This paper proposes a novel adaptive framework for processing dynamic signal flows over simplicial complexes, extending classical least-mean-squares (LMS) strategies to high-order topological domains. Leveraging theoretical insights from discrete Hodge theory, we introduce a topological LMS algorithm that effectively gathers information from streaming signals, which are observed over time-varying subsets of edges. We present a comprehensive stochastic analysis of the algorithm, deriving its stability conditions, mean-square-error at steady state, and convergence speed. Through this investigation, we highlight the effect of edge sampling on the algorithm's learning capability and steady-state performance. Then, we illustrate principled strategies to design the sampling probability on each edge of the complex, with the aim of minimizing the sampling rate while maintaining desired estimation performance. Assuming partial knowledge of the simplicial complex structure (e.g., only the underlying graph is known), we also propose an adaptive strategy to perform topology inference, which integrates the previous method. Finally, a distributed variant of the topological algorithm is also proposed, along with a detailed study of its stability and mean-square-error properties. The theoretical results are supported by empirical experiments on synthetic and real-world traffic data, which show that our algorithm, both in the centralized and distributed cases, outperforms graph-based LMS methods, by incorporating higher-order topological features.
\end{abstract}

\begin{IEEEkeywords}
Adaptive learning, topological signal processing, simplicial complexes, diffusion strategies, topology inference.
\end{IEEEkeywords}

\section{Introduction}
% 2 pages
Combinatorial topological spaces have been increasingly used to model and process complex systems such as social, biological, and communication networks \cite{battiloro_tesi,lambiotte2019networks}. Graphs are the simplest and most studied combinatorial topological spaces. In a graph,  pairwise interactions are encoded through graph edges \cite{Huang}. Consequently, a wide range of methods for analyzing signals over graphs has emerged, forming the framework of Graph Signal Processing (GSP) \cite{ortega2018graph}.  Important achievements in GSP include the introduction of graph convolutional filters to process signals and the definition of the graph Fourier transform, which offers a consistent generalization of the traditional discrete Fourier transform to irregular discrete domains \cite{isufi2024graph}. Many forms of relations in complex systems, however, are not just binary, but involve sets of more than two elements and thus can not be effectively represented through graph edges \cite{schaub2021signal}. Significant examples are given by biological neural networks, in which simultaneous activaction of group of neurons are relevant in the dynamics of the brain \cite{lambiotte2019networks} and social networks, where multiple relationships between individuals are the norm \cite{Ghahremani}. In these cases, graph representations are not exhaustive, since multi-way relations between vertices necessitate different topological descriptors to be encoded. To this end, signal processing tools have been extended to analyze signals over higher-order topological domains such as hypergraphs, simplicial complexes and cell complexes, often referred to as \textit{topological signals}, giving rise to the emergent field of Topological Signal Processing (TSP) \cite{barbarossa2020topological,ghorbanchian2021highord,Roddenberry}. Simplicial complexes, in particular, by definition have a rigid algebraic structure, which is completely described by a finite family of incidence matrices. A further incentive for simplicial-based representations is the solid algebraic theoretical foundation associated with these topological spaces, given by discrete Hodge theory, which provides topology-aware linear operators to process signals \cite{barbarossa2020topological}. Typically, the definition of these operators depends on the so-called Hodge Laplacian matrices, which extend the traditional graph Laplacians to higher dimensional spaces \cite{goldberg2002combinatorial}. 

Recent works in TSP have addressed classical problems of signal theory such as representation, inference and filtering of signals defined over simplicial complexes. In particular, in \cite{topological_slepians} the authors propose dictionary building strategies for sparse representation of topological signals. The work in \cite{barbarossa2020topological} includes techniques for recovering signals over simplicial complexes from partial information and novel topology inference methods to estimate the structure of an unknown simplicial complex from edge observations. An innovative notion of duality for topological signal spaces is discussed in \cite{barbarossa_2}, with the introduction of frequency domains and Fourier transforms to alternatively represent signals. Subsequently, simplicial convolutional filters have been developed \cite{SCF_9893391}, to regulate frequency responses and to extract relevant components in topological signals. Recently, a dictionary learning framework based on topological convolutional filters has been proposed, to enable sparse representation of signals over regular cell complexes \cite{grimaldi2025topologicaldictionarylearning}.
%\textcolor{red}{Add citation to Topological Dictionary Learning journal.}  
Furthermore, neural network architectures based on simplicial complexes have been proposed \cite{papillon2024,battiloro2024generalized}, to address learning tasks over topological spaces \cite{bodnar2021weisfeiler,SAT,Chen_power_grid} and latent topology inference \cite{battiloro2023latentgraphlatenttopology}. However, the study of signals on simplicial complexes in the works mentioned above pertains to \textit{static} signals, i.e. signals without temporal evolution. On the other hand, observations of flow signals in complex systems are often more appropriately modeled through temporal sequences of topological signals. This highlights the need to develop TSP techniques designed for \textit{dynamic} data, to address tasks such as learning or inference over topological spaces.

In this perspective, a fundamental area of investigation is adaptive learning, which is the renowned field in signal processing that encompasses the set of statistical learning techniques that leverage continuously updated information about a system to gain insights into its structure and/or evolution over time \cite{sayed2011adaptive,sayed2014adaptation}. The uniqueness of these methods lies in their \textit{online} approach, which offers advantages in terms of flexibility in learning when the dynamics of a system change over time \cite{sayed2022inference}. Adaptive methodologies for graph signal processing were proposed in a series of works. As an example, the authors in \cite{Graph_inference} propose an adaptive algorithm to infer the latent graph structure of a system from dynamic signals observations, whose evolution over time is regulated by heat diffusion laws. Traditional online estimation methods, such as least mean squares, recursive least squares (RLS), and Kalman flitering have been reinterpreted to address signal recovery tasks in the GSP framework \cite{di2016adaptive,isufi2020observing}. Instead, the work in \cite{di2018adaptive} introduces adaptive methods for the reconstruction of graph signals from observations collected over time-varying, randomly selected sets of nodes. Additionally, the works in \cite{di2017distributed,nassif2018distributed,hua2020online} employ distributed adaptive processing techniques on streaming data to identify graph filters.  
%However, in hte current literature, most of adaptive methodologies over topological spaces are primarily tailored for graph signals.
%, and can not be effectively implemented over discrete domains with more sophisticated structures such as simplicial or cell complexes. 

%The literature on topological adaptive signal processing is still relatively limited and to date, only a few works have addressed topological signals. Kalman filtering approaches for simplicial-based processes are discussed in \cite{Money, Money_kalman}. In detail, authors in \cite{Money} focus on edge flow imputation, exploiting flow conservation priors induced by the topology. In \cite{Money_kalman}, a topology-aware Kalman filter is introduced, which models network dynamics using simplicial convolutional filters. Authors in \cite{yang2023online} propose an online method for edge flow prediction on simplicial complexes with time-varying topology. Instead, the work in \cite{Yan2025} proposed an adaptive method over simplicial complexes to jointly estimate vertex and edge signals.
The literature on topological adaptive signal processing remains relatively limited, with only a few works addressing signals defined over simplicial complexes \cite{isufi2025topologicalsignalprocessinglearning}. Kalman filtering approaches have been explored in \cite{Money, Money_kalman}: the work in \cite{Money} focuses on edge flow imputation by leveraging flow conservation priors induced by the topology, while \cite{Money_kalman} introduces a topology-aware Kalman filter that models network dynamics through simplicial convolutional filters. An online method for edge flow prediction over time-varying simplicial complexes is proposed in \cite{yang2023online}. More recently, \cite{Yan2025} introduces an adaptive approach for jointly estimating vertex and edge signals over simplicial complexes. The closest related work is presented in \cite{Krishnan2024}, which introduced an adaptive learning framework for simplicial complexes. However, there are several key differences between this work and ours. First, while \cite{Krishnan2024} focuses on simplicial vector autoregressive models, our primary goal is to extend adaptive LMS algorithms to topological domains, along with their theoretical analysis. Second, our proposed methods incorporate a probabilistic sampling mechanism that can be tailored to improve the performance of the algorithms. Third, the methods in \cite{Money,Money_kalman,yang2023online,Yan2025,Krishnan2024} require prior knowledge of the complex’s topology, whereas in Section IV, we introduce an adaptive approach that infers the topology while simultaneously performing the learning task. Finally, the approaches in \cite{Money,Money_kalman,yang2023online,Yan2025,Krishnan2024} are inherently centralized, whereas we present a scalable, distributed processing strategy that enables adaptive learning from streaming topological signals.\\
\textbf{Contributions.} The goal of this paper is to introduce a novel adaptive framework for analyzing and processing dynamic signal flows defined over simplicial complexes. The main contributions of this work can be summarized as: \vspace{-.02cm}
\begin{itemize}
    \item Hinging on theoretical results from discrete Hodge Theory, we develop a new topological LMS algorithm that leverages dynamic edge signal observations to perform adaptive topological system identification. Specifically, we extend our preliminary work in \cite{marinucci2024topological}, generalizing the graph-based adaptive learning method in \cite{nassif2018distributed} by incorporating higher-order topological features, to process and learn from streaming signals defined over simplicial complexes. Additionally, through an in-depth stochastic analysis, we study the stability conditions of the proposed algorithm and derive closed-form formulas for its mean-square-error at steady state and convergence speed.
    \item We develop a \textit{probabilistic} sampling mechanism that explores the trade-off between sampling rate and performance of the algorithm. Specifically, the proposed strategy minimizes the sampling rate while guaranteeing steady-state performance in terms of mean-square-error and convergence speed. Furthermore, we assess the validity of the sampling methods through convergence analysis and numerical simulations.
    \item We propose an adaptive LMS technique to estimate the latent simplicial-complex structure of a system. Assuming the graph-level connections to be known (or previously estimated), higher-order topological features are learned jointly with the generating process of the observed data. 
    \item Finally, a distributed version of the topological algorithm is introduced, designed for implementation across a network of agents. We also outline the stability conditions, deriving explicit expressions of the mean-square behavior, influenced by the underlying topological structure, sampling techniques and data characteristics. Also in this case, empirical evaluations validate our theoretical findings.
\end{itemize}
\noindent \textbf{Outline.} The paper is organized as follows. Section II reviews key tools in topological signal processing. Section III introduces the topological LMS algorithm, analyzes its asymptotic behavior, and proposes principled sampling strategies. Section IV focuses on adaptive topology inference. Section V presents the distributed implementation and derives stability conditions similar to the centralized case. Section VI provides numerical validation on synthetic and real traffic data, comparing our method with the graph-based LMS from \cite{nassif2018distributed}. Section VII draws the conclusions of the paper.\smallskip

\noindent \textbf{Notation.} Scalar, column vector, and matrix variables are indicated by plain letters $a$, bold lowercase letters $\mathbf{a}$, and bold uppercase letters $\mathbf{A}$, respectively. $[\mathbf{a}]_i$ is the $i$-th element of vector $\mathbf{a}$, $[\mathbf{A}]_{i,j}$ is the $(i,j)$-th element of $\mathbf{A}$, $[\mathbf{A}]_i$ is the $i$-th row of $\mathbf{A}$, $[\mathbf{A}]^i$ is the $i$-th column of $\mathbf{A}$, and $\lambda_{MAX}(\mathbf{A})$ denotes the largest eigenvalue of $\mathbf{A}$.  $\mathbf{I}$ is the identity matrix, $\mathbf{1}$ is the vector of all ones, and $\mathbf{1}_{\mathcal{S}}$ is a vector with ones at the indices in $\mathcal{S}$ and zeros elsewhere. $\textrm{im}(\cdot)$, $\textrm{ker}(\cdot)$, and $\text{supp}(\cdot)$ denote the image, the kernel, and the support of a matrix, respectively; $\oplus$ is the direct sum of vector spaces, and $\otimes$ denotes the Kronecker product. $\{\mathbf{a}_k\}_{k=1}^K$ and $\{\mathbf{A}_k\}_{k=1}^K$ are the collection of $K$ vectors and matrices, respectively. Other specific notation is defined along the paper, if needed.

\section{Background}
In this section, we review the Topological Signal Processing (TSP) concepts required to introduce the adaptive filtering framework. Please refer to \cite{barbarossa2020topological, Hatcher} for further details.

\subsection{Simplicial Complexes and Signals} 

\noindent\textbf{Simplicial complexes.} A (finite abstract) \textit{simplicial complex} (SC) $\mathcal{X}$ is a pair $(\mathcal{S}, \mathcal{V})$, where $\mathcal{V}$ is a finite set of vertices and $\mathcal{S}$ is a family of subsets of $\mathcal{V}$ such that: (i) For every $x \in \mathcal{V}$, $\{x\} \in \mathcal{S}$; (ii) if $\sigma_i \in \mathcal{S}$ and $\sigma_j \subset \sigma_i$, then $\sigma_j \in \mathcal{S}$.
A set of $k+1$ vertices respecting the properties above is called a \textit{k-simplex}, or simplex of order $k$, and we denote it with $\sigma_i^k$. The \textit{dimension} of a SC $\mathcal{X}$ is the highest order of its simplices. A face of a $k$-simplex $\sigma_i^k$ is a subset of it with cardinality $k$, thus a $k$-simplex has $k+1$ faces. 

\noindent\textbf{Topological Signals.} Let us denote the set of $k$-simplex in $\mathcal{X}_{K}$ as  ${\cal D}_{k} := \{\sigma_{k,i}: \sigma_{k,i} \in \mathcal{X}_{K}\} $, with cardinality $|{\cal D}_{k}| = N_k$. A $k$-simplicial signal is defined as a collection of mappings from the set of all $k$-simplices contained in the complex to real numbers:
\begin{equation}\label{signals}
    \mathbf{x}_k = [x_k(\sigma^k_1),\dots,x_k(\sigma^k_i), \dots, x_k(\sigma^k_{N_k})]^T \in \mathbb{R}^{N_k},
\end{equation}
where $x_{k}: {\cal D}_{k} \rightarrow \mathbb{R}$. In most of the cases the focus is on complexes of dimension two, thus having a set of vertices $\mathcal{V}$ with $|\mathcal{V}| = V$, a set of edges $\mathcal{E}$ with $|\mathcal{E}|=E$, and a set of triangles $\mathcal{T}$ with $|\mathcal{T}| = T$, which result in ${\cal D}_{0}={\cal V}$ (simplices of order 0), ${\cal D}_{1}={\cal E}$ (simplices of order 1), and ${\cal D}_{2}={\cal T}$ (simplices of order 2). In this case, the $k$-simplicial signals are defined by the following mappings: 
\begin{equation}
    x_{0}: {\cal V} \rightarrow \mathbb{R} , \qquad x_{1}: {\cal E} \rightarrow \mathbb{R} , \qquad x_{2}: {\cal T} \rightarrow \mathbb{R} ,
\end{equation}
leading to graph, edge, and triangle signals, respectively.
\noindent\textbf{Adjacencies of Simplices.} We say that two $k$-simplices are \textit{upper adjacent} if they are both faces of a common $(k+1)$-simplex. Similarly, two $k$-simplices are \textit{lower adjacent} if they share a common face, i.e. a $(k-1)$-simplex. As such, $0$-simplices (vertices) can only be upper adjacent, recovering the usual node adjacency in a graph, while $K$-simplices can only be lower adjacent. As illustrated in Fig. \ref{Fig:adjacencies}, two edges are lower adjacent if they share a common endpoint node, while they are upper adjacent if they are both part of the same triangle.

\subsection{Algebraic Representation and Filters}
As with graphs, each simplex can be assigned one of two possible \textit{orientations}, depending on the ordering given to its vertices. More specifically, two different orderings of the vertices of a simplex induce the same orientation on it if one can be obtained by the other with an even number of transpositions, otherwise they induce opposite orientations. Additionally, given a $(k-1)$-simplex $\sigma_i^{k-1}$ and a $k$-simplex $\sigma_j^{k}$, such that $\sigma_i^{k-1} \subset \sigma_j^{k}$, they are \textit{coherently oriented} (denoted with $\sigma_i^{k-1}\sim\sigma_j^{k}$), if the orientation induced by $\sigma_j^{k}$ to its faces is in accordance with the orientation of $\sigma_i^{k-1}$, or \textit{uncoherently oriented} (denoted with $\sigma_i^{k-1}\nsim\sigma_j^{k}$) otherwise \cite{Hatcher}. 

%An \textit{oriented simplicial complex} of dimension $K$ is just a simplicial complex in which each simplex is equipped with an orientation. Signals on these topological spaces may be observed at various levels, or more precisely on simplices of different orders.\\

\noindent\textbf{Incidence Matrices.} Given a reference orientation, a simplicial complex of dimension $K$ is entirely determined by the set of its \textit{incidence matrices} $\mathbf{B}_{k}$, $k=1,...,K$, which encode the inclusion relations between the $(k-1)$-simplices and the $k$-simplices in the complex: 
\begin{equation}\label{eq:incidence}
  \big[\mathbf{B}_{k} \big]_{i,j}=\left\{\begin{array}{rll}
  0, & \text{if} \; \sigma_i^{k-1} \not\subset \sigma_j^{k}; \\
  1,& \text{if} \; \sigma_i^{k-1} \subset \sigma_j^{k} \;  \text{ and $\sigma_i^{k-1}\sim\sigma_j^{k}$;}\\
  -1,& \text{if} \; \sigma_i^{k-1} \subset \sigma_j^{k} \;  \text{  and $\sigma_i^{k-1}\nsim\sigma_j^{k}$.}\\
  \end{array}\right.
  \end{equation} 

  \begin{figure}[t]
    \centering
    \includegraphics[width = \linewidth]{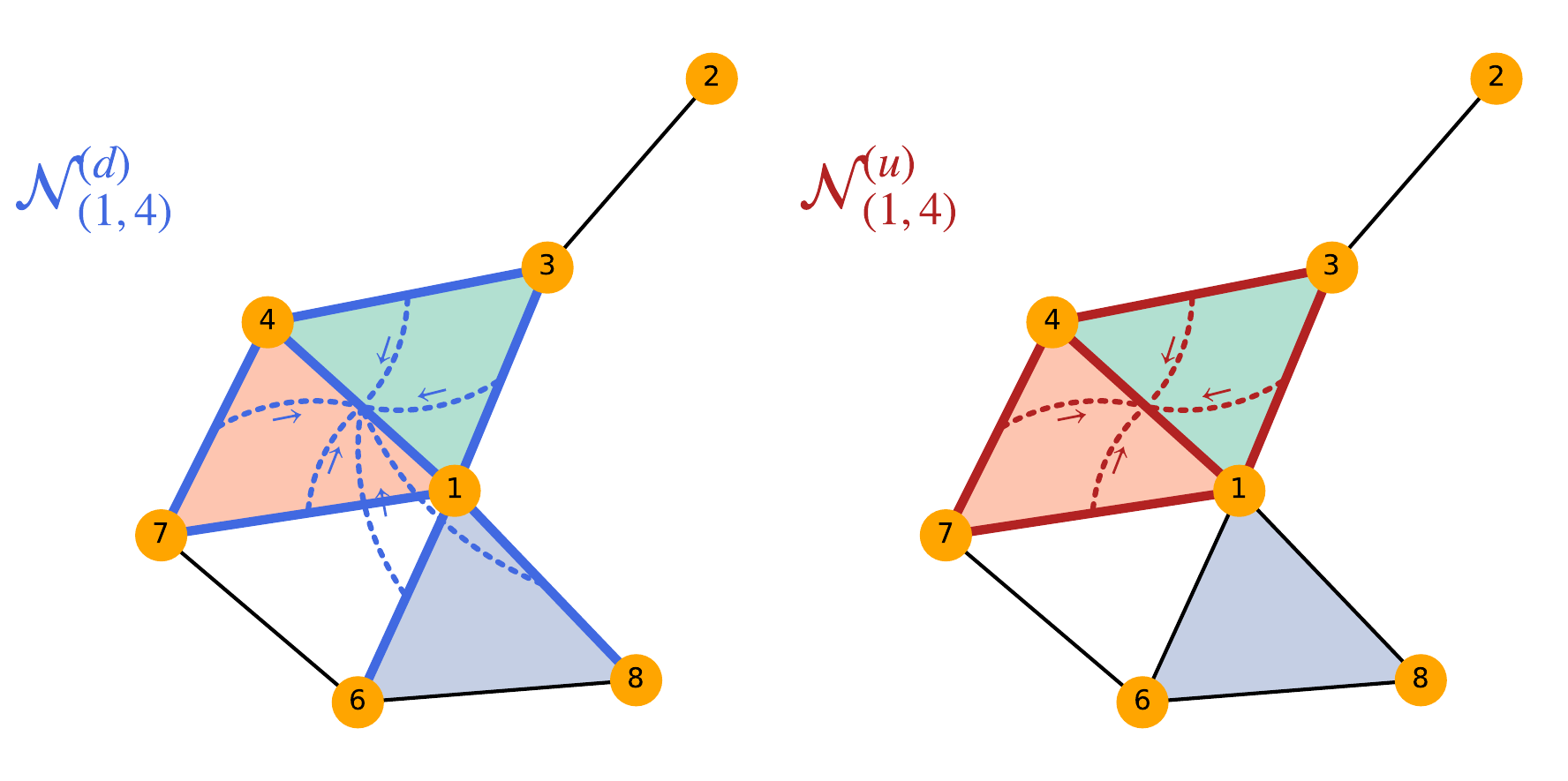}
    \caption{\textit{(Left)} The edge $(1,4)$ is lower-adjacent to the blue edges (denoted with $\mathcal{N}^{(d)}_{(1,4)})$.
    \textit{(Right)} The edge $(1,4)$ is upper-adjacent to the red edges (denoted with $\mathcal{N}^{(u)}_{(1,4)})$.\vspace{-.5cm}}
    \label{Fig:adjacencies}
\end{figure}

\noindent \textbf{Hodge Laplacians.} From the incidence information, it is possible to define the \textit{Hodge Laplacian} matrices as follows.
\begin{align}
 &\mathbf{L}_0 =\mathbf{B}_1\mathbf{B}^T_1 \notag\\
 &\mathbf{L}_k = \mathbf{B}^T_k\mathbf{B}_k+\mathbf{B}_{k+1}\mathbf{B}^T_{k+1}, \quad k=1,...,K-1 \\
 &\mathbf{L}_K = \mathbf{B}^T_K\mathbf{B}_K \notag
\end{align}
Note that each Laplacian of intermediate order, $k=1,...,K-1$ has two components, called the \textit{upper Laplacian} $\mathbf{L}_{k,u}=\mathbf{B}_{k+1}\mathbf{B}^T_{k+1}$ and the \textit{lower Laplacian} $\mathbf{L}_{k,d}=\mathbf{B}^T_{k}\mathbf{B}_{k}$. The upper Laplacian of order $k$ encodes the \textit{upper adjacency} between $k$-simplices, i.e. $[\mathbf{L}_{k,u}]_{i,j}\neq 0$ iff $\sigma^k_i$ and $\sigma^k_j$ are upper adjacent.  Similarly, the lower  Laplacian of order $k$ encodes the \textit{lower adjacency} between $k$-simplices, i.e. $[\mathbf{L}_{k,d}]_{i,j}\neq 0$ iff $\sigma^k_i$ and $\sigma^k_j$ are lower adjacent. $\mathbf{L}_0$ recovers the usual graph Laplacian.
\noindent\textbf{Hodge Decomposition.} Hodge Laplacians admit a Hodge decomposition, stating that the signal space associated with each simplex of order $k$ can be decomposed as the direct sum of the following three orthogonal subspaces \cite{lim2020hodge}:
\begin{equation} \label{hodge_spaces}
\mathbb{R}^{N_{k}} = \text{im}(\mathbf{B}_{k}^T\big) \oplus \text{im}\big(\mathbf{B}_{k+1}\big) \oplus \text{ker}\big(\mathbf{L}_{k}\big).
\end{equation}
Thus, every signal $\mathbf{x}_{k}$ of order $k$ can be decomposed as:
\begin{equation}
\label{hodge_decomp}
\mathbf{x}_{k}=\mathbf{B}_{k}^T\, \mathbf{x}_{k-1}+\mathbf{B}_{k+1}\, \mathbf{x}_{k+1}+\widetilde{\mathbf{x}}_{k},
\end{equation}
where $\widetilde{\mathbf{x}}_{k} \in \text{ker}\big(\mathbf{L}_{k}\big)$.

\noindent \textbf{Simplicial Fourier Transform} As discussed in  \cite{barbarossa2020topological},  Hodge Laplacians encode local and global topological properties of the complex, and its eigenvectors form a principled basis for topological signals. Let us denote the eigendecomposition  of the Hodge Laplacian of order $k$ with $\mathbf{L}_k=\mathbf{U}_k\Lambda_k\mathbf{U}^T_k$. The \textit{Simplicial Fourier transform} (SFT) $\hat{\mathbf{x}}_k$ of a $k$-topological signal $\mathbf{x}_k$  is defined as the projection of $\mathbf{x}_k$ over the basis of eigenvectors of $\mathbf{L}_k$ \cite{barbarossa2020topological}:
\begin{equation}\label{eq:sft}
\hat{\mathbf{x}}_k=\mathbf{U}^T_k\mathbf{x}_k 
\end{equation}
%To extract frequency components of topological signals, simplicial convolutional filters have been introduced. These operators linearly process signals and act like point-wise multiplications in the frequency domains, in accordance with the general convolutional theorem \cite{SCF_9893391}.
%An important result in algebraic topology, known as the discrete Hodge decomposition (discussed in \cite{barbarossa2020topological}), shows that signals belonging to different eigenspaces of the Laplacians exhibit a dissimilar differential behavior (in a discrete sense). eigenvectors of Hodge Laplacians encompass global features about the structure of the simplicial complex.

\noindent \textbf{Simplicial convolutional FIR filters.} Given the SFT definition in \eqref{eq:sft} and the Hodge decomposition in \eqref{hodge_spaces}, a \textit{simplicial convolutional FIR filter} $\H_k$ acting on $k-$topological signals is defined as a matrix polynomial of the upper $k$-Hodge Laplacian $\mathbf{L}_{k,u}$ and lower $k$-Hodge Laplacian $\mathbf{L}_{k,d}$ \cite{yang2023convolutional}, i.e.,
\begin{equation}\label{sep_mat}
    \H_k = \sum_{m=0}^{M}h_{m,u}\big(\L_{k,u}\big)^m+\sum_{m=1}^{M}h_{m,d}\big(\L_{k,d}\big)^m \notag,
\end{equation}
where $M$ is a positive integer and $\{h_{m,u}\}_{m=0}^M $, $\{h_{m,d}\}_{m=1}^M$ are the filter coefficients. In other words, a simplicial convolutional filter is simply obtained as a summation, moduled by scalar coefficients, of iterated \textit{shift} operations, 
which combine the signal values of \textit{upper} and \textit{lower} neighboring simplices. W.l.g., from now on we assume $\mathcal{X}$ to be a $2$-dimensional simplicial complex and we focus on edge signals $\mathbf{x}_1$.  For this reason, in the sequel, we will simply indicate $\mathbf{L}_{1,u}$ and $\mathbf{L}_{1,d}$ with $\mathbf{L}_u$ and $\mathbf{L}_d$, respectively, and $\mathbf{x}_1$ with $\mathbf{x}$.

\begin{remark}
    All TSP concepts introduced in this section naturally extend to regular cell complexes. In particular, the notions of lower and upper adjacency, as well as their representation with the lower and upper Laplacians $\L_d$ and $\L_u$, % and consequently the FIR topological filters,  
     are defined analogously to the simplicial case. As a result, the theoretical findings of this work are directly applicable to cell complexes as well. Indeed, initial formulations of LMS algorithms on cell complexes have already been discussed in \cite{marinucci2024topological}, further supporting the applicability of this framework beyond simplicial settings.
\end{remark}

\section{Topological Adaptive Least Mean Squares}
In this section, we introduce an LMS algorithm for adaptive filtering of streaming and partially observed topological signals over simplicial complexes. Let $\mathcal{X}$ be a 2-dimensional simplicial complex, and let $\mathbf{x}(n)$ be a stationary edge signal processed over time by the FIR filters in \eqref{sep_mat} according to the model:
\begin{align}\label{eq:lin_obs}
\boldsymbol{y}(n)=\mathbf{D}(n) & {\left[\sum_{m=0}^M h_{m, u}^o\left(\mathbf{L}_u\right)^m \mathbf{x}(n-m)\right.} \nonumber\\
&\hspace{-1.3cm} \left.+\sum_{m=1}^M h_{m, d}^o\left(\mathbf{L}_d\right)^m \mathbf{x}(n-m)+\boldsymbol{v}(n)\right], \quad n \geq M,
\end{align}
where $\mathbf{D}(n)={\rm diag}(d_1(n),\ldots,d_E(n))\in\mathbb{R}^{E\times E}$ is a random, time-varying sampling matrix, such that $d_i(n)=1$ if the signal's component on the edge $i$ is observed at time $n$, and zero otherwise; $\bv(n)= \left[ v_1(n), \ldots,v_E(n)\right]^T \in \mathbb{R}^{E}$ is an i.i.d. zero-mean measurement noise, independent of $\{\bx(n)\}_{n \in \mathbb{N}}$ and $\{\mathbf{D}(n)\}_{n \in \mathbb{N}}$, with diagonal covariance matrix $\mathbf{C}_v$; $\mathbf{h}^o=[h^o_{u,0},h^o_{1,u},\ldots,h^o_{M,u},h^o_{1,d},\ldots,h^o_{M,d}]^T\in\mathbb{R}^{2M+1}$ is the vector of filter coefficient that we aim to estimate and track. Intuitively, the observation model in (\ref{eq:lin_obs}) involves signal $\bx(n)$ and its previous observations, which are linearly processed by a simplicial convolutional filter that operates by combining neighboring edge signals in the upper and lower sense, through the actions of $\mathbf{L}_u$ and $\mathbf{L}_d$, respectively. Our model in \eqref{eq:lin_obs} is a non-trivial generalization for topological signals of the graph-based model in \cite{nassif2018distributed}, in which vertex signals are observed and processed by shift-invariant graph filters. For our framework and the related analysis, we assume the following.

\textit{Assumption 1 (Independent sampling):} The random variables $d_i(l)$ extracted from the sampling process are spatially and temporally independent, for all $i = 1, ..., E$, and $l \leq n$.

\textit{Assumption 2 (Stationarity):} The signal observations $\bx(n)$ arise from a zero-mean, wide-sense stationary process, i.e., \(\mathbb{E}\{\bx(n)\} = \mathbf{0}\) for all $n$, and its autocorrelation sequence $\mC_x(m)=\mathbb{E}\{\bx(n)\bx^T(n - m)\}$ does not depend on the time index $n$, being a function of the time lag \(m\) only.

\subsection{Algorithm Derivation}

For the sake of exposition, let us define the time-varying matrix $\bX(n) \in\mathbb{R}^{E \times 2M+1}$ as
\begin{align}\label{eq:X}
    \bX(n) = \big[\bx(n),& \L_u \bx(n - 1),\ldots, \L_u^{M}\bx(n - M), \nonumber\\
    &\quad \L_d \bx(n - 1),\ldots, \L_d^{M}\bx(n - M)\big],
\end{align}
which collects the aggregated regressors in (\ref{eq:lin_obs}) as columns. Note that the shift operators $\L_u$ and $\L_d$ only combine observations on neighboring edges, in the upper and lower sense respectively. Therefore, the $i-$th row of $\bX(n)$ can be computed locally for each edge $i \in E$, since it depends on edges at most $M$ (upper and lower) hops away in the complex. 

We employ the mean square error criterion \cite{sayed2011adaptive}, i.e., the estimates of the filter coefficients $\mathbf{h}^o$ are obtained by minimizing 
\begin{align}\label{eq:MSE}
   J(\mathbf{h}):=\mathbb{E}\left\{\big\|\by(n) - \D(n)\bX(n)\mathbf{h}\big\|^2 \right\}.
\end{align}
Given the convexity of $J$, the optimal $\mh$ is found by setting the gradient of $J$ equal to zero, i.e., solving the normal equations:
\begin{align}\label{eq:normaleq}
   \mC_X\mh = \mathbf{c}_{Xy}, 
\end{align}
where the matrix \(\mC_X\)$\in\mathbb{R}^{2M+1\times 2M+1}$ and the vector \(\mathbf{c}_{Xy}\)$\in\mathbb{R}^{2M+1 \times 1}$ are given by:
\begin{align}
 &\mC_X = \mathbb{E}\{\mathbf{X}(n)^T\D(n)\mathbf{X}(n)\},  \label{eq:covariance}\\ 
 &\mathbf{c}_{Xy} = \mathbb{E}\{\mathbf{X}(n)^T\D(n)\by(n)\}. \nonumber
\end{align}
Note that $\mC_X$ and $\mathbf{c}_{Xy}$ do not depend on $n$, since $\bx(n)$ and $\D(n)$ are stationary by Assumption 2. The system of equations in \eqref{eq:normaleq} admits a unique solution if the square matrix $\mC_X$ is positive definite, i.e., $\lambda_{\min}(\mC_X)>0$, where $\lambda_{\min}(\mY)$ is the minimum eigenvalue of matrix $\mY$. As (\ref{eq:covariance}) suggests, this condition depends crucially on the sampling process. Indeed, it can be shown by direct substitution that the matrix \(\mC_X\) has a $2\times 2$ block structure:
\begin{equation}\label{eq:Rx}
\mC_X=\begin{bmatrix}
\mC_{u,X} & \mC_{ud,X} \\
\mC_{ud,X} & \mC_{d,X}
\end{bmatrix},    
\end{equation}
and the matrices $\mC_{u,X},\mC_{d,X}$ and $\mC_{ud,X}$ can be evaluated explicitly. In particular, let $\mP$ denote the \textit{expected} sampling matrix $\mP = \mathbb{E}\{\D(n)\}$. Therefore, the $i$-th diagonal entry of $\mP$ is the probability $\mathbb{E}\{d_i(n)\}=p_i$ of sampling the $i$-th edge. The $(m,l)$ entry of matrix $\mC_{u,X} $ is then given by
\begin{align}\label{eq:RuX}
    [\mC_{u,X}]_{m,l} =&\; \mathbb{E}\left\{\bx^T(n - m)(\L_u^m)^T \D(n) \L_u^{l}\bx(n - l)\right\} \nonumber\\
     =& \;\text{Tr}\left((\L_u^{m})^T \mP \L_u^{l}\mC_x(m - l)\right).
\end{align}
Similarly, it holds:
\begin{align}\label{eq:RdX}
    &[\mC_{d,X}]_{m,l} =\text{Tr}\left((\L_d^{m})^T \mP \L_d^{l}\mC_x(m - l)\right),\nonumber\\
    &[\mC_{ud,X}]_{m,l} =\text{Tr}\left((\L_u^{m})^T \mP \L_d^{l}\mC_x(m - l)\right).
\end{align}
Finally, $\mathbf{c}_{Xy}=[\mathbf{c}^u_{Xy},\mathbf{c}^d_{Xy}]$ can be recast as
\begin{align}
    &\mathbf{c}^u_{Xy}=\text{Tr}\left((\L_u^{m})^T \mP \mC_{xy}(m)\right), \nonumber\\
    &\mathbf{c}^d_{Xy}=\text{Tr}\left((\L_d^{m})^T \mP \mC_{xy}(m)\right).
\end{align}
In the above equations, $\mC_{xy}(m)=\mathbb{E}\{\by(n)\bx^T(n-m)\}$ is the cross-correlation function, again independent of the time index $n$ because of Assumptions 1 and 2. Interestingly, from (\ref{eq:RuX})-(\ref{eq:RdX}), the positive definiteness of $\mC_X$ depends on a proper design of the sampling probability vector $\mathbf{p} ={\rm diag}(\mP)= (p_1, \ldots, p_E)^T$, %, and consequently of the probability matrix $\mP={\rm diag}(\bp)$. 
which will be elaborated in Sec. III.\textit{B}. 

If $\mC_X$ is positive definite, the optimal parameter vector \(\mh^o\) can be estimated by either solving \eqref{eq:normaleq} or adopting iterative approaches. A straightforward option is updating the  estimates of the filter coefficient by gradient descent, i.e.,
\begin{align}\label{eq:grad_desc}
   \mh(n + 1) = \mh(n) + \mu \big( \mathbf{c}_{Xy} - \mathbf{C}_X \mh(n) \big),
\end{align}
with $\mu>0$ being a (sufficiently small) step-size. Nevertheless, having access to the second order moments of the observed signals is an unrealistic assumption in practical applications; thus, the matrix $\mathbf{C}_X$ and the vector $\mathbf{c}_{Xy}$ usually can not be evaluated beforehand. Using an LMS adaptive approach \cite{sayed2011adaptive}, we can exploit their instantaneous approximations:
\begin{align}
    \mathbf{C}_X \approx \mathbf{X}^T(n)\D(n) \mathbf{X}(n), \quad \mathbf{c}_{Xy} \approx \mathbf{X}^T(n) \D(n)\by(n).
\end{align}
This results in the following stochastic, adaptive filter:
\begin{align}\label{eq:grad_desc}
   \mh(n + 1) = \mh(n) + \mu \mathbf{X}^T(n)\D(n)\big( \by(n) - \mathbf{X}(n) \mh(n) \big). 
\end{align}
We refer to this approach as the topological LMS algorithm (Topo-LMS), which processes streaming topological signals $\{\by(n),\bx(n)\}$ by leveraging the structure of the underlying simplicial complex. The main steps of the proposed Topo-LMS method are listed in Algorithm 1.

\begin{algorithm}
\caption{: Topo-LMS on Simplicial Complexes}
\begin{algorithmic}[1]
\State \textbf{Data:} Set $\mh(0)$ randomly; choose $\mu>0$.
\For{each time $n \geq 0$}
    \State Choose a sampling operator $\D(n)$;
    \State Observe $\by(n)$, and build $\bX(n)$ as in (\ref{eq:X});
    %\State $   \bh(n + 1) = \bh(n) + \mu \mathbf{X}^T(n)\D(n)\big( \by(n) - \mathbf{X}(n) \bh(n) \big)$
    \State Perform parameter adaptation as:
    \begin{align}
   \quad\mh(n + 1) = \mh(n) + \mu \mathbf{X}^T(n)\D(n)\big( \by(n) - \mathbf{X}(n) \mh(n) \big) \nonumber
\end{align}
\EndFor
\end{algorithmic}\label{alg:topolms}
\end{algorithm} 

\subsection{Stochastic Behavior}
We now investigate the mean and mean-square behavior of Algorithm 1. We define the error vector at time $n$ as: 
\begin{equation}
\widetilde{\mh}(n)=\mh^o-\mh(n) .    
\end{equation} 
According to expressions (\ref{eq:lin_obs}) and (\ref{eq:grad_desc}), its evolution over time satisfies the following recursion:
\begin{align}\label{eq:error_recursion}
    \widetilde{\mh}(n+1)=\mathbf{Q}(n)\widetilde{\mh}(n)-\mu \bg(n)
\end{align}
where
\begin{align}
    &\mathbf{Q}(n)=\mathbf{I}-\mu \mathbf{X}(n)^T\D(n)\mathbf{X}(n), \\
    & \bg(n)=\mathbf{X}(n)^T\D(n)\bv(n).
\end{align}
In the following, we introduce an assumption of temporal independence of the regression matrices $\{\mathbf{D}(n)\bX(n)\}_{n \in \mathbb{N}}$ to make the analysis more tractable \cite{sayed2011adaptive,sayed2014adaptation}. 

\textit{Assumption 3 (Independent regressors):} The regressors $\mathbf{D}(n) \mathbf{X}(n)$ arise from a zero mean random process that is temporally white with positive definite covariance $\mathbf{C}_X$, i.e., $\mathbf{D}(n)\mathbf{X}(n)$ is independent of $\mathbf{D}(l)\mathbf{X}(l)$, for every $l \leq n$. 

Even if Assumption 3 is not strictly verified in model (\ref{eq:lin_obs}) , it is widely used to study the converge of adaptive methods \cite{sayed2011adaptive,nassif2018distributed}.
A direct consequence of Assumption 3 is that $\mathbf{Q}(n)$ is independent of $\widetilde{\mh}(n)$ for every $n$. Thus, since $\mathbb{E}\{\bg(n)\} = \mathbf{0}$, taking the expected value on each side of (\ref{eq:error_recursion}), we get:
\begin{equation}\label{eq:expected_error}
    \mathbb{E}\{\widetilde{\mh}(n+1)\} = \mathbf{Q} \,\mathbb{E}\{\widetilde{\mh}(n)\},
\end{equation}
with $\mathbf{Q}=\mathbf{I}-\mu \mathbf{C}_X$. From (\ref{eq:expected_error}), the convergence in mean of the Topo-LMS algorithm  depends on the stability of matrix $\mathbf{Q}$. Specifically, the recursion (\ref{eq:grad_desc}) converges in mean to the optimal solution if the spectral radius $\rho(\mathbf{Q}) < 1$, or, equivalently, if the step-size $\mu$ is chosen to satisfy:
\begin{equation}\label{eq:step}
 0 < \mu < \frac{2}{\lambda_{\max}(\mathbf{C}_{X})}. 
 \end{equation}

In the sequel, we illustrate how condition (\ref{eq:step}) induces mean-square convergence of Algorithm 1 to a limit value, for which we provide a closed-form expression. To this end, we study the evolution of $\mathbb{E}\{\|\widetilde{\mh}(n)\|_{\boldsymbol{\Sigma}}^2\}$ for an arbitrary positive semi-definite matrix $\boldsymbol{\Sigma}$. Evaluating the weighted variance on both sides of (\ref{eq:error_recursion}) under Assumptions 1-3, we get:
\begin{equation}\label{eq:var_relation}
\mathbb{E}\{\|\widetilde{\mathbf{h}}(n+1)\|_{\boldsymbol{\Sigma}}^2\} = \mathbb{E}\{\|\widetilde{\mathbf{h}}(n)\|_{\boldsymbol{\Sigma}'}^2\}+ \mu^2\text{Tr}(\boldsymbol{\Sigma} \mathbf{G}),
\end{equation}
with $\boldsymbol{\Sigma}'  = \mathbb{E}\{\mathbf{Q}^T(n)\boldsymbol{\Sigma}\mathbf{Q}(n)\}$, and
\begin{align}\label{eq:covariance_g} 
\mathbf{G} &= \mathbb{E}\{\mathbf{X}(n)^T\mathbf{C}_{v}\mP\mathbf{X}(n)\}.
\end{align}
%using the spatially independence assumption of the zero-mean noise $\bv(n)$. 
Note that matrix $\mathbf{G}$ in (\ref{eq:covariance_g}) can be formulated explicitly, following a procedure similar to \eqref{eq:RuX}-\eqref{eq:RdX}. Let us  define the vectors $\boldsymbol{\boldsymbol{\sigma}} = \text{vec}(\boldsymbol{\Sigma})$ and $\boldsymbol{\sigma}' = \text{vec}(\boldsymbol{\Sigma}')$, obtained by vectorizing $\boldsymbol{\Sigma}$ and $\boldsymbol{\Sigma}'$, respectively. Given 
\begin{equation} \label{eq:F_matrix}
\mathbf{F} = \mathbb{E}\{\mathbf{Q}^T(n) \otimes \mathbf{Q}^T(n)\},    
\end{equation}
we have that $\boldsymbol{\sigma}' = \mathbf{F}\boldsymbol{\sigma}$. The definition of $\mathbf{F}$ in (\ref{eq:F_matrix}) involves the fourth-order moments of the flow signals, which can not be derived under the current assumptions.
To overcome this issue, an often used approximation of $\mathbf{F}$ is given by $\mathbf{Q}^T \otimes \mathbf{Q}^T$, which is an accurate alternative when the step-size $\mu$ is sufficiently small, so that the terms involving higher-order powers of the $\mu$ can be ignored \cite{sayed2011adaptive,sayed2014adaptation}. In this case, the stability of $\mathbf{F}$ is ensured if $\rho(\mathbf{Q}) < 1$, i.e., if the step-size satisfies (\ref{eq:step}). 

Denoting the weighted norm $\|\cdot \|_{\boldsymbol{\Sigma}}$ as $\|\cdot \|_{\boldsymbol{\sigma}}$ with a slight abuse of notation, the variance evolution in (\ref{eq:var_relation}) can be rewritten as:
\begin{equation}\label{eq:var_relation2}
\mathbb{E}\{\|\widetilde{\mathbf{h}}(n+1)\|_{\boldsymbol{\sigma}}^2\} = \mathbb{E}\{\|\widetilde{\mathbf{h}}(n)\|_{\boldsymbol{\mathbf{F}\sigma}}^2\}+ \mu^2\text{vec}\left(\mathbf{G}\right)^T \boldsymbol{\sigma}.
\end{equation}
It is straightforward that this recursion can be recast as:
\begin{equation}\label{variance_iter}
\mathbb{E}\{\|\widetilde{\mathbf{h}}(n+1)\|^2_{\boldsymbol{\sigma}}\} = \mathbb{E}\{\|\widetilde{\mathbf{h}}(0)\|^2_{\mathbf{F}^n\boldsymbol{\sigma}}\}+ \mu^2\text{vec}\left(\mathbf{G}\right)^T \sum_{i=0}^{n}\mathbf{F}^i\boldsymbol{\sigma}.
\end{equation}
Therefore, the weighted variance $\mathbb{E}\{\|\widetilde{\mathbf{h}}(n)\|^2_{\boldsymbol{\sigma}}\}$ converges to the limit value:
\begin{equation}\label{eq:limit-point}
\lim_{n \to \infty} \mathbb{E}\left\{\|\widetilde{\mathbf{h}}(n)\|_{\boldsymbol{\sigma}}^2  \right\} = \mu^2\text{vec}\left(\mathbf{G}\right)^T  (\mathbf{I} - \mathbf{F})^{-1}\boldsymbol{\sigma}. 
\end{equation}
From (\ref{eq:limit-point}), we can obtain the mean-square deviation (MSD) at steady-state by simply choosing $\boldsymbol{\sigma}=\text{vec}(\mathbf{I})$.

We conclude noting that the transient component in (\ref{variance_iter}) vanishes as $\mathbf{F}^n$ for $n \xrightarrow{} \infty$. Thus, the convergence speed to the limit value can be measured through the norm $\alpha=\|\mathbf{F}\|_2$, which we refer to as \textit{convergence rate}. In particular, a smaller value of the convergence rate $\alpha$ corresponds to a faster mean-square convergence of the Topo-LMS algorithm to steady state.

\subsection{Optimal Sampling Strategies}
%It is suggested by stochastic analysis in Sec. III.\textit{A} that sampling probability distribution influences the Topo-LMS algorithm's performance in terms of MSD at steady state and convergence rate. Deeply inspired by the work in \cite{di2018adaptive}, this subsection is therefore dedicated to further exploring the relationship between sampling probabilities and these two performance indicators, in order to 
%establish \textit{where} and at what frequency is it advisable to sample signal observations over the complex.  
It is clear from \eqref{eq:limit-point} that the sampling probability distribution affects the Topo-LMS algorithm's learning performance. For this reason, in this section, we formally investigate the relationship between sampling probabilities, MSD at steady state, and convergence rate, ultimately  
establishing  optimal sampling strategies. To this aim, in the following theorem, we first elaborate on \eqref{eq:limit-point} to obtain expressions for MSD and convergence rate in which the dependence on the sampling probability vector $\mathbf{p}=\text{\rm diag}(\mP)$ becomes more explicit.
\begin{Theorem}\label{theo:samp_1}
Under Assumptions 1-3, if condition (\ref{eq:step}) holds, the MSD at steady state of the Topo-LMS algorithm can be written as:
\begin{equation}\label{eq:limit_var}
{\rm MSD}=\lim_{n \to \infty} \mathbb{E}\{\|\widetilde{\mathbf{h}}(n)\|^2\}=\frac{\mu}{2}\,{\rm Tr}\left(\mathbf{G}(\mathbf{p})\mathbf{C}_X^{-1}(\mathbf{p})\right) + \textit{O}({\mu}^2) , 
\end{equation}
where 
\begin{align}
        \mathbf{G}(\mathbf{p}) &=  \mathbb{E}\{ \mathbf{X}(n)^T {\rm diag (\mathbf{p})} \mathbf{C}_v \mathbf{X}(n)\}, \notag \\
      \mathbf{C}_X(\mathbf{p}) &= \mathbb{E}\{ \mathbf{X}(n)^T {\rm diag}(\mathbf{p}) \mathbf{X}(n)\}.\label{eq:Cx}
\end{align}
Furthermore, letting $\delta=\lambda_{\text{min}}\left(\mathbf{C}_X(\mathbf{p})\right)$ and $\nu=\lambda_{\text{max}}\left(\mathbf{C}_X(\mathbf{p})\right)$, if $ \mu \ll 2 \delta/\nu^2$, the convergence rate $\alpha$ of the Topo-LMS algorithm can be well approximated as:
\begin{equation}\label{eq:convergence_rate}
\alpha \simeq 1-2\mu \lambda_{\rm min}(\mathbf{C}_X(\mathbf{p})).    
\end{equation}
\end{Theorem}
\begin{proof}
See Appendix A in Supplementary Material.
\end{proof}

% From now on, we omit the term $O(\mu^2)$ in (\ref{eq:limit_var}) by % relying on the assumption 
% assuming that the step-size $\mu$ is chosen sufficiently small, so that terms depending on higher-order powers of $\mu$ can be ignored.

% \begin{Theorem}\label{theo:samp_2}
% Let $\delta=\lambda_{\text{min}}\left(\mathbf{C}_X\right)$ and $\nu=\lambda_{\text{max}}\left(\mathbf{C}_X\right)$. Under Assumption 1-2, if $ \mu << 2 \delta/\nu^2$, the convergence rate $\alpha$ of the Topo-LMS algorithm can be approximated as
% \begin{equation}\label{eq:convergence_rate}
% \alpha \simeq 1-2\mu \lambda_{\rm min}(\mathbf{C}_X)    
% \end{equation}
%  with
%  \begin{align}
%       \mathbf{C}_X &= \mathbb{E}\{ \mathbf{X}(n)^T \mP \mathbf{X}(n)\}, \quad  \mP = \text{diag}(\mathbf{p}). 
% \end{align}
% \end{Theorem}
% \begin{proof}
% See Appendix.
% \end{proof}

In the sequel, based on the findings of Theorem 1, we introduce an optimization strategy
aimed at minimizing the probability to sample simplices under learning performance guarantees, i.e., under constraints on the MSD in (\ref{eq:limit_var}) and convergence rate in (\ref{eq:convergence_rate}).  Formally, the optimal sampling probability vector $\mathbf{p}$ is determined as the solution of: 
\begin{align}\label{prob}
    &\underset{\mathbf{p}}{\text{min}} \quad 1^T\mathbf{p}  \notag\\ 
    &\quad \text{subject to:} \\
    &\qquad \textrm{(a)} \;\; 0 \leq \mathbf{p} \leq \mathbf{p}_{\text{max}} \nonumber\\ 
&\qquad \textrm{(b)} \;\;\lambda_\text{min}\left(\mathbf{C}_X(\mathbf{p})\right) \geq \frac{1-\alpha}{2\mu} \nonumber \\
&\qquad \textrm{(c)} \;\;\text{Tr}\left( \mathbf{G}(\mathbf{p})(\mathbf{C}_X(\mathbf{p}))^{-1} \right) \leq \frac{2\gamma}{\mu} \nonumber
\end{align}
where constraint (a), in which the $\leq$ is element-wise, imposes non-negativity and upper bounds on $\mathbf{p}$, which are well motivated by external factors that often occur in practical applications (e.g. limited energy or communication resources);  constraint (b) ensures that the convergence rate is at least $\alpha$ according to expression \eqref{eq:convergence_rate}; finally, constraint (c) ensures that the MSD is smaller than $\gamma$ according to \eqref{eq:limit_var}.

Unfortunately, problem \eqref{prob} is non-convex because of constraint (c). To overcome this issue, we propose a convex relaxation of Problem \eqref{prob}, where constraint (c) is replaced with a convex one. In particular, we employ the upper bound
\begin{equation}
  \text{Tr}\left( \mathbf{G}(\mathbf{p})(\mathbf{C}_X(\mathbf{p}))^{-1} \right) \leq \frac{\text{Tr}\left(\mathbf{G}(\mathbf{p}) \right)} {\lambda_\text{min}\left(\mathbf{C}_X(\mathbf{p})\right)} ,
\end{equation}
which holds true because, under constraint (b) in \eqref{prob}, $\mathbf{G}(\mathbf{p})$ and $\mathbf{C}_X(\mathbf{p})$ are positive definite matrices \cite{Sheng-DeWang}.
%To achieve this goal, we design a constrained optimization problem whose solution minimizes the sampling rate, under constraints on the MSD and convergence rate. 
% In problem \eqref{prob}, Additionally, from \eqref{eq:step}, %it is worth noting that 
% the constraint (b) imposes the non-singularity of matrix $\mathbf{C}_X$, which guarantees the algorithm's convergence in mean to the optimal solution $\bh^o$. Constraint (a), in which the $\leq$ is element-wise, imposes the positivity and upper bounds of the probabilities in $\mathbf{p}$ and it is motivated by external factors that often occur in applications (e.g. limited energy or communication resources). \\
%Unfortunately, problem \eqref{prob} is not convex because constraint (c) is not convex, thus we cannot use efficient numerical optimization methods to solve it.
% To overcome this issue, we propose a convex relation of Problem \eqref{prob}, where constraint (c) is replaced with a convex one. In particular, we employ the upper bound
% \begin{equation}
%   \text{Tr}\left( \mathbf{G}(\mathbf{p})(\mathbf{C}_X(\mathbf{p}))^{-1} \right) \leq \frac{\text{Tr}\left(\mathbf{G}(\mathbf{p}) \right)} {\lambda_\text{min}\left(\mathbf{C}_X(\mathbf{p})\right)} ,
% \end{equation}
% holding because $\mathbf{G}$ and $\mathbf{C}_X$ are positive definite matrices \cite{Sheng-DeWang}. 
Therefore, the convex relaxation of Problem \eqref{prob} reads as:
\begin{align}\label{prob_2}
    &\underset{\mathbf{p}}{\text{min}} \quad 1^T\mathbf{p}  \notag\\ 
    &\quad \text{subject to:} \\
    &\qquad \textrm{(a)} \;\; 0 \leq \mathbf{p} \leq \mathbf{p}_{\text{max}} \nonumber\\ 
&\qquad \textrm{(b)} \;\;\lambda_\text{min}\left(\mathbf{C}_X(\mathbf{p})\right) \geq \frac{1-\alpha}{2\mu} \nonumber \\
&\qquad \textrm{(c)} \;\;\frac{\text{Tr}\left(\mathbf{G}(\mathbf{p}) \right)} {\lambda_\text{min}\left(\mathbf{C}_X(\mathbf{p})\right)}  \leq \frac{2\gamma}{\mu}\nonumber
\end{align}
The new constraint (c) in Problem \eqref{prob_2} is given by the ratio between a convex and a concave function, both positive and differentiable in the prescribed region. This defines a pseudo-convex map \cite{Avriel}, whose sub-level sets are convex. %In combination with t
The convexity of the other constraints and the linearity of the objective function make Problem \eqref{prob_2} convex, and thus solvable using efficient numerical methods \cite{CVX}.

\section{Adaptive Topology Inference}
The framework outlined in the previous section is tailored for simplicial complexes with a \textit{known} and \textit{static} topology that remains fixed over time. This knowledge is encoded into the structure of the upper and lower Laplacian matrices $\mathbf{L}_u$ and $\mathbf{L}_d$ in (\ref{eq:lin_obs}). However, in several applications (e.g. wireless networks, social networks, functional brain networks) the topology is often latent or may change dynamically over time. In this section, we introduce a new procedure for \textit{joint} adaptive topology learning and filtering of streaming and partially observed topological signals over simplicial complexes, building on Algorithm (\ref{alg:topolms}). In the following, we assume that the topology of the complex is known up to order $k$ when processing $k$-signals. For example, if we are working with edge signals, we assume that nodes and edges are provided, while higher-order structures such as triangles (i.e., the upper adjacency of edges) must be inferred from the data. The topology inference task is then equivalent to inferring the upper Laplacian of the simplicial order we work on. In the edge case, let $T_{\textrm{max}}$ be the number of cliques of three elements of the underlying graph, i.e., the candidate triangles, and $T\leq T_{\textrm{max}}$ be the number of actual triangles. The upper Laplacian $\L_u \in \mathbb{R}^{E \times E}$ encodes the presence of triangles among all the 3-cliques, because it can be written as: 
\begin{equation}\label{eq:Laplacian_decomp}
    \L_u = \mathbf{B}_2\mathbf{B}_2^T =\sum_{j=0}^{T_{\textrm{max}}} t^o_j \mathbf{b}_j \mathbf{b}_j^T,
\end{equation}
where each term in the sum corresponds to a candidate triangle, $t^o_j \in \{0,1\}$ is a binary variable indicating whether the $j$-th 3-clique forms a triangle ($t^o_j = 1$) or not ($t^o_j = 0$), and $\mathbf{b}_j \in \mathbb{R}^E$ encodes the incidence relations between the $j$-th 3-clique and the edges of the complex. In other words, if $t^o_j = 1$, then $\mathbf{b}_j$ is a column of the incidence matrix $\mathbf{B}_2 \in \mathbb{R}^{E \times T}$. 

Using (\ref{eq:Laplacian_decomp}) into the signal model \eqref{eq:lin_obs}, we obtain:
\begin{align}
    \by(n) = \D(n)&\Bigg[\sum_{m=0}^{M}h^o_{m,u}\left(\sum_{j=0}^T t^o_j \mathbf{b}_j \mathbf{b}_j^T\right)^m \bx(n - m) \nonumber\\
       &\hspace{-.5cm}+ \sum_{m=1}^{M}h^o_{m,d}\big(\L_d\big)^m \bx(n - m) +\bv(n)\Bigg],    
\end{align}
where $\h^o$ (i.e., the filter coefficients) and $\mathbf{t}^o = [t_1,\dots,t_{T_{\textrm{max}}}] \in \mathbb{R}^{T_{\textrm{max}}}$ (i.e., the triangle indicators) are the quantities we aim to estimate in a joint and adaptive fashion. Analogously to \eqref{eq:X}, we define the time-varying matrix $\bX(\mathbf{t}, n) \in \mathbb{R}^{E \times (2M+1)}$, which now depends on the topology-related vector variable $\mathbf{t}$. As before, we use the mean-square criterion, now jointly estimating the filter coefficients and triangle indicators by solving
\begin{align}\label{top_inference}
    &\min_{\h,\,\mathbf{t}} \quad  
    J(\h,\mathbf{t})\,=\,\mathbb{E}\left\{\left\|\by(n) - \D(n)\bX(\mathbf{t},n)\h\right\|^2 \right\}\\
    &\qquad \text{subject to} \quad  \mathbf{t} \in \{0,1\}^{T_{\textrm{max}}} \ . \nonumber
\end{align}
% where each term of the sum is associated with a candidate triangle, $t^o_j \in \{0,1\}$ is a binary variable indicating if the $j$-th 3-clique is a triangle ($t^o_j=1)$ or not ($t^o_j=0$),  and $\mathbf{b}_j \in \mathbb{R}^E$ describes the incidence relations between the $j$-th 3-clique and the edges of the complex. 
% as:
% \begin{align}\label{eq:X2}
%     \bX(n,\bt) = \Bigg[\bx(n),& \ldots, \left(\sum_{j=0}^{T_{\textrm{max}}} t_j b_j b_j^T\right)^{M}\bx(n - M), \nonumber\\
%     &\quad \ldots, \L_d^{M}\bx(n - M)\Bigg].
% \end{align}
% \begin{align}\label{top_inference}
%     J(\h,\mathbf{t}):=\mathbb{E}\{\big\|\by(n) - \D(n)\bX(\mathbf{t},n)\bh\big\|^2 \} 
%     \textrm{ s.t. }\bt \in \{0,1\}^{T_{\textrm{max}}} \ .
% \end{align}
Unfortunately, problem \eqref{top_inference} is non-convex and NP-hard. To address this challenge, we propose a continuous relaxation of the constraint combined with a two-step stochastic gradient descent algorithm. Specifically, we begin by relaxing the binary constraint, allowing the vector $\mathbf{t}$ to take continuous values in the interval $[0,1]$. To encourage sparse and near-binary solutions, we then introduce attractor functions that promote values of $\mathbf{t}$ close to 0 or 1. Specifically, we rewrite the problem as: 
\begin{align}\label{top_inference_relaxed}
    &\min_{\h,\,\mathbf{t}} \quad  
    J(\h,\bt) + \lambda_0\|\mathbf{t}\|_0 + \lambda_1 \|\mathbf{t}-\mathbf{1}\|_0\\
    &\qquad \text{subject to} \quad  \mathbf{t} \in [0,1]^{T_{\textrm{max}}}, \nonumber
\end{align}
with $J(\h,\mathbf{t})$ defined in (\ref{top_inference}), and $\lambda_0,\lambda_1 \geq 0$ denoting positive regularization parameters that can control the weight of the zero-attracting and one-attracting terms in (\ref{top_inference_relaxed}), respectively. Then, we exploit employ a \textit{two-step} stochastic projected gradient descent procedure, alternatively updating the filter coefficients and the triangle indicators estimates as:
% \begin{align}
%         &\h(n + 1) = \h(n) - \mu_1\nabla_{\h} J(\h(n),\mathbf{t}(n))\label{eq:proxgd_h}\\
%         &\mathbf{t}(n+1) = \operatorname{prox}_{\lambda_0}^{\lambda_1} \Big[ \mathbf{t}(n) - \mu_2 \nabla_{\mathbf{t}} J(\h(n+1),\mathbf{t}(n)) \Big]\label{eq:proxgd_t}
% \end{align}
\begin{align}
    \begin{aligned}
        &\h(n + 1) = \h(n) - \mu_1\nabla_{\h} J(\h(n),\mathbf{t}(n))\label{eq:proxgd}\\
        &\mathbf{t}(n+1) = \mathcal{H}_{\lambda_0}^{\lambda_1} \Big[ \mathbf{t}(n) - \mu_2 \nabla_{\mathbf{t}} J(\h(n+1),\mathbf{t}(n)) \Big]
    \end{aligned}
\end{align}
with $\mu_1, \mu_2>0$ being (sufficiently small) step-sizes, and $\mathcal{H}_{\lambda_0}^{\lambda_1}$ denoting the proximal hard-thresholding operator obtained as:
\begin{align}
\mathcal{H}_{\lambda_0}^{\lambda_1}(\mathbf{v}) =
\underset{\mathbf{u} \in [0,1]^{T_{\textrm{max}}}}{\arg\min}  &\frac{1}{2} \|\mathbf{u} - \mathbf{v}\|^2 
+ \lambda_0 \|\mathbf{u}\|_0 \notag \\
&+ \lambda_1 \|\mathbf{u} - \mathbf{1}\|_0. \label{eq:prox}
\end{align}
The closed form expression for (\ref{eq:prox}) is given in the sequel.
\begin{lemma}
The proximal operator $\mathcal{H}_{\lambda_0}^{\lambda_1}$ in (\ref{eq:prox}) reads as:
\begin{equation}\label{eq:Double_Hard_Thresh}
\mathcal{H}_{\lambda_0}^{\lambda_1}(v_i) =
\begin{cases} 
0, & v_i \leq \sqrt{2\lambda_0} \\
v_i, & \sqrt{2\lambda_0} < v_i < 1 - \sqrt{2\lambda_1} \\
1, & v_i \geq 1 - \sqrt{2\lambda_1} \ ,
\end{cases}
\end{equation}
for all $i=1,\ldots,E$.
\end{lemma}
\begin{proof}
See Appendix B in Supplementary Material.
\end{proof}
Intuitively, the proximal hard-thresholding operator in (\ref{eq:Double_Hard_Thresh}) encourages sparsity by setting to zero all entries smaller than $\sqrt{2\lambda_0}$, while simultaneously promoting activation (i.e., setting to one) for entries larger than $1 - \sqrt{2\lambda_1}$, thus favoring the presence of active triangles. From (\ref{eq:Double_Hard_Thresh}), it is also clear that the parameters $\lambda_0$ and $\lambda_1$ must be chosen to satisfy the inequality $1 - \sqrt{2\lambda_1}>\sqrt{2\lambda_0}$. Also in this case, we adopt a stochastic approach where the gradients $\nabla_{\h} J(\h(n),\mathbf{t}(n))$ and $\nabla_{\mathbf{t}} J(\h(n+1),\mathbf{t}(n))$ in (\ref{eq:proxgd}) are substituted by instantaneous approximations based on online data observations. The resulting algorithm (\ref{eq:proxgd}) filters streaming topological signals $\{\by(n),\bx(n)\}$ while learning the latent topology of the underlying simplicial complex. We will illustrate its performance and adaptation capabilities in Sec. VI.

\section{Distributed Topological Least-mean Squares}
%In the previous section, the proposed Topo-LMS 
Up to now, we have considered algorithms operating in a centralized fashion, %characterized by the
i.e, they implicitly assume that all the data is gathered by a single processing unit running the adaptive filtering. This approach can lead to bottlenecks and single-point vulnerabilities as system scale increases \cite{sayed2014adaptation,di2018adaptive,nassif2018distributed,di2020distributed}. In distributed systems, individual agents instead collect and process their local measurements and can only share information with nearby agents. In our topological setting, the agents are the edges of the simplicial complex aiming to perform the adaptive filtering in a decentralized and cooperative way, and they exchange information according to two distinct connectivity schemes encoded by the lower and upper adjacencies induced by the simplicial complex.
%\textcolor{blue}{In this work, we consider agents on the edges of the simplicial complex who communicate directly through the lower connectivity, i.e. with agents located on edges incident to a common vertex. Nonetheless, employing communication schemes based on upper connectivity, or a hybrid lower-upper approach, would produce equivalent results.} 
Real-world examples of this setting can be found in several applications related to the monitoring and control of critical network infrastructure, such as traffic, hydraulic, or communication networks, where flow measurements (i.e., edge signals) play a crucial role. 

In this section, we  propose a distributed version of Topo-LMS (cf. Algorithm \ref{alg:topolms}) using an \textit{adaptive diffusion} approach \cite{cattivelli2009diffusion,nassif2018distributed}, promoting local consensus and spread of information over the simpicial complex. In this distributed setting, the evolution of the signal's component $y_i(n)$ over the $i$-th edge can be easily written according to \eqref{eq:lin_obs} as: 
    \begin{align}\label{eq:lin_obs2}
    y_i(n) = &\;d_i(n)\sum_{m=0}^{M}h^o_{m,u}\big[\big(\L_u\big)^m \bx(n - m)\big]_i \nonumber\\
       &\hspace{-1.2cm}+ d_i(n)\sum_{m=1}^{M}h^o_{m,d}\big[\big(\L_d\big)^m \bx(n - m)\big]_i +d_i(n)v_i(n)    %\\ 
  %     = & \bz_i(n)^T h_0 + +\big[\D(n)\bv(n)\big]_i 
\end{align}
for all $i=1,\ldots,E$, where $v_i(n)$ denotes the $i$-th component of the noise vector $\bv(n)$. Let us introduce the vector
\begin{equation}\label{z_i}
    \bz_i(n):= [\bX(n)]_{i}^T,
\end{equation}
with $\bX(n)$ given by (\ref{eq:X}). Then, the observation (\ref{eq:lin_obs2}) can be equivalently cast as:
\begin{align}\label{eq:lin_obs3}
    y_i(n) = d_i(n)\left(\bz_i(n)^T \h^o  + v_i(n)\right), \quad i=1,\ldots,E.  
\end{align} 
Importantly, from (\ref{z_i}) and (\ref{eq:X}), the regressor $\bz_i(n)$ can be computed locally at edge $i$ by combining information from neighboring edges, as illustrated in Fig. \ref{Fig:adjacencies}. This is achieved diffusing signals over the simplicial complex through the successive application of the upper and lower Laplacians, $\L_u$ and $\L_d$, over their respective neighborhoods. Now, the MSE in \eqref{eq:MSE} can be rewritten as a sum over the edges, letting to the  distributed learning problem:
\begin{align}\label{distributed_problem}
    &\underset{\h_i}{\text{min}} \ \sum_{i=1}^E \mathbb{E}\{|y_i(n) - d_i(n)\bz_i(n)^T\h_i|^2 \} \\ 
    &\text{subject to:} \quad \h_i = \h_j, \; \hbox{for all $i,j\in\mathcal{N}_i$,} \nonumber
\end{align}
where $\{\h_i\}_{i=1}^E$ are local copies of the global variable $\h$, constrained to reach asymptotic consensus. In (\ref{distributed_problem}), $\mathcal{N}_i$ represents the set of neighbors of edge $i$ used to diffuse information and reach consensus over the copies. Note that this further communication graph is not necessarily coincident with the one induced by $\mathbf{L}_d$ or $\mathbf{L}_u$, and represents a further degree of freedom in the design of the distributed algorithm. Now, following the diffusion adaptation approach from \cite{cattivelli2009diffusion,sayed2014adaptation,nassif2018distributed}, we proceed by tailoring the \textit{Adapt-then-Combine $(ATC)$} distributed algorithm to our topological learning problem in (\ref{distributed_problem}), thus obtaining the recursions:
\begin{align}
&\hspace{-.1cm}\bw_i(n+1) = \h_i(n)+\mu_i d_i(n)\bz_i(n) (y_i(n)-\bz_i(n)^T \h_i(n))  \medskip \nonumber\\
&\h_i(n+1) = \displaystyle\sum\limits_{l\in \mathcal{N}_i}a_{i,l}\,\bw_l(n+1) \label{eq:combine}
\end{align}
for all $i = 1,\ldots, E$, where 
%$\mathcal{N}_{i}$ is a set of neighbors of the $i$-th edge used to reach consensus (not necessarily coincident with the sets induced the upper or lower connectivity of the complex), 
$\{\mu_i\}_{i=1}^E$ are positive step-sizes, and $\{a_{i,l}\}$ are local combination coefficients satisfying \cite{cattivelli2009diffusion}:
\begin{equation}
    a_{i,l}\geq 0, %\text{ for } l=1,...,N
    \quad a_{i,l} = 0 \ \text{  if  }  \ l\notin \mathcal{N}_{i}, \quad \sum\limits_{l\in \mathcal{N}_{i}}a_{i,l}=1 \ .
    \label{eq:filtering conditions}
\end{equation} 
In \eqref{eq:combine}, if the sampling variable $d_i(n)=1$, each agent first adapts its estimate $\h_i(n)$ using its own data $\{\by_i(n),\bz_i(n)\}$, which can be built via data exchange with neighbor agents; then, each local intermediate estimate $\bw_i(n+1)$ is combined with the intermediate estimates of neighboring edges with weights $\{a_{k,l}\}$ to update the estimates $\h_i(n+1)$ and push them toward consensus. If the sampling variable $d_i(n)=0$, the agent performs only the combination step, diffusing local information over the network. The main steps of the proposed distributed method are summarized in Algorithm \ref{alg:topolms_distr}.

\subsection{Stochastic Behavior}
In this section, we study the stability conditions of Algorithm \ref{alg:topolms_distr}, and its convergence properties in terms of mean and mean-square behavior. Let $\mathbf{A}=\{a_{k,l}\}$, and let $\widetilde{\mathbf{h}}_{\text{dec}}(n)$ be the network error vector, which stacks the estimation error of each edge agent, namely:
\begin{align}\label{eq:network_error}
    \widetilde{\mathbf{h}}_{\text{dec}}(n)=\text{col}\left\{\h^o-\mathbf{h}_i(n)\right\}^E_{i=1}.
\end{align}
From (\ref{eq:combine}) and (\ref{eq:network_error}), the time evolution of $\widetilde{\mathbf{h}}_{\text{dec}}(n)$ reads as:
\begin{align}\label{eq:error_recursion2}
    \widetilde{\mathbf{h}}_{\text{dec}}(n+1)=\mathbf{B}(n)\widetilde{\mathbf{h}}_{\text{dec}}(n)-\mathcal{A}\mathbf{M} \bg(n)
\end{align}
where
\begin{align}
    &\mathbf{B}(n)=\mathcal{A}(\mathbf{I}_{E\cdot (2M+1)}-\mathbf{M} \mathbf{C}_z(n)) \notag \\
    &\mathbf{M}= \text{diag}\{\mu_i \mathbf{I}_{2M+1}\}_{i=1,...,E} \notag\\
       &\mathcal{A}= \mA \otimes \mathbf{I}_{2M+1} \notag 
\end{align}
\begin{align}   
   &\mathbf{C}_z(n)= \text{diag}\{d_i(n)\bz_i(n)\bz_i(n)^T\}_{i=1,...,E} \ \notag \\
    & \bg(n)= \text{col}\{d_i(n)\bz_i(n)v_i(n)\}_{i=1,...,E} \notag \ 
\end{align}
Similarly to the centralized case, we introduce this assumption.

\begin{algorithm}
\caption{: Distributed Topo-LMS}
\begin{algorithmic}[1]
\State \textbf{Data:} Set $\h(0)$ randomly, choose $\{a_{k,l}\}$ satisfying (\ref{eq:filtering conditions}), set  $\{\mu_i\}>0$.
\For{each time $n \geq 0$}
    \State  Choose sampling variables $d_1(n),\ldots,d_E(n)$;
    \For{each edge $i \in E$}
    \State Observe $y_i(n)$, and build $\bz_i(n)$ as in (\ref{z_i});
    \State Perform distributed adaptation as:
    \begin{align}\label{eq:combine}\hspace{-.2cm}
     &\bw_i(n+1) = \h_i(n)+\mu_id_i(n)\bz_i(n) (y_i(n)-\bz_i(n)^T \h_i(n)) \nonumber \medskip\\
     &\h_i(n+1) = \displaystyle\sum\limits_{l\in \mathcal{N}_i}a_{i,l}\,\bw_l(n+1)\nonumber
     \end{align}
    %\State Retain the vector $\tilde{\bz}_i(n-1)$ as in (\ref{z_l})
    \EndFor
\EndFor
\end{algorithmic}\label{alg:topolms_distr}
\end{algorithm}

\textit{Assumption 4 (Local independence):} The regressors $\bz_i(n)$ arise from a zero-mean random process that is temporally white.

Assumption 4 directly implies that $\mathbf{B}(n)$ is independent of $\widetilde{\h}_{\text{dec}}(n)$ for every $n$, %. Despite this assumption might not be true in general, it 
and is commonly used when analyzing adaptive estimation algorithms %since it allows 
to simplify the derivations without constraining the conclusions \cite{sayed2011adaptive}. Then, as a consequence of Assumption 4 and %the fact that 
since $\mathbb{E}\{\bg(n)\} = \mathbf{0}$, taking the expectation of both sides of (\ref{eq:error_recursion2}), we obtain:
\begin{equation}\label{eq:mean_recursion_dist}
    \mathbb{E}\{\widetilde{\mathbf{h}}_{\text{dec}}(n+1)\} = \mathbf{B} \,\mathbb{E}\{\widetilde{\mathbf{h}}_{\text{dec}}(n)\} ,
\end{equation}
with
%\begin{eqnarray}
%    \mathbf{Q}&=&\mathcal{A}(I_{E\cdot (2M+1)}-\mathcal{M} \mathbf{C}) \notag \\
%   \mathbf{C}&=& \text{diag}\{\mathbf{C}_{X,k}\}_{k=1,...,E} \ .\notag 
%\end{eqnarray}
\begin{align}
  &\mathbf{B}=\mathcal{A}(\mathbf{I}_{E\cdot (2M+1)}-\mathbf{M} \mathbf{C}_z), \label{eq:Q_matrix}\\
  &\mathbf{C}_z = \mathbb{E}\{\mathbf{C}_z(n)\} = \text{diag}\{\mathbf{C}_{z,i}\}_{i=1,...,E} \label{eq_Cz}
\end{align}
where $\mathbf{C}_{z,i}=\mathbb{E}\{d_i(n)\bz_i(n)\bz_i(n)^T\}$. Thus, from (\ref{eq:mean_recursion_dist}), the estimates $\{\mathbf{h}_i(n)\}_{i=1}^E$ generated by (\ref{eq:combine}) converge in mean to the ground truth vector $\h^o$ if $\mathbf{B}$ is stable, i.e., if $\rho(\mathbf{B}) < 1$. 

In Sec. III, to ensure the mean convergence of the Topo-LMS algorithm, we assumed the non-singularity of the covariance matrix $\mathbf{C}_X$ in (\ref{eq:Cx}). Naturally, one could extend the assumption in this context to all the local covariance matrices $\mathbf{C}_{z,i}$ in (\ref{eq_Cz}), for all $i=1,\ldots,E$; this would simply guarantee, as for the centralized method, convergence in mean of Algorithm (\ref{alg:topolms_distr}), if the step-sizes satisfy local stability conditions as in (\ref{eq:step}). This assumption is typical in distributed learning, and can be found in several previous works, see, e.g., \cite{cattivelli2009diffusion,nassif2018distributed}.
It is easy to observe that such an assumption, for signals defined on simplicial complexes, imposes very restrictive conditions on the structure of the complex itself, particularly requiring it to satisfy completeness hypotheses. Specifically, each edge must be the face of some triangle in the simplicial complex. %In fact, to ensure that the model in (\ref{eq:lin_obs3}) is non-degenerate for every $k\in E$, it is necessary that the set of edges upper-adjacent to each edge $k$ is non-empty. %
In fact, if this were not the case for some edge \( j \in E \), the \( j \)-th row of the matrix \( \mathbf{L}^u \) would be zero. As a result, according to the definition of \( z_j \) given in (\ref{z_i}), its corresponding entries would be identically zero. Consequently, matrix $\mathbf{C}_{z,j}$ would be singular, violating the assumption. 
%\textcolor{red}{Explain what you mean by completeness, and why it is required.} \\
Moreover, from a learning standpoint, since adaptation on each edge  $k$ depends directly on the matrix  $\mathbf{C}_{z,i}$, assuming these matrices are non-singular implies that each local agent could asymptotically learn the true parameter $\h^o$ independently, without requiring communication or diffusion across the network. However, this scenario is generally unlikely due to the limited local knowledge of each edge’s complex structure and the presence of the sampling operator, which further constrains data observability.

In the following, we show that the method converges in the mean under milder conditions, requiring only one agent to have a non-singular matrix \( \mathbf{C}_{z,i} \), provided the communication topology and filtering coefficients \( \{a_{i,l}\} \) are appropriately designed. In particular, we require that the aggregate coefficient matrix \( \mA \) is \emph{irreducible}, meaning that any two agents in the network can communicate within a finite number of steps, as specified in the following definition \cite{norris1998markov}.
\begin{definition}
A stochastic matrix $\mA$ is irreducible if, for every pair of states $i$ and $j$, there exists a positive integer $k$ such that the 
$(i,j)$ element of the matrix 
$\mA^k$ is positive, i.e., $(\mA^k)_{i,j}>0$. This means that it is possible to reach state $j$ from state 
$i$ in a finite number of steps.
\end{definition}

The convergence in mean of algorithm (\ref{eq:combine}) is then ensured by the following theorem.

\begin{Theorem}\label{distr}
Assume that \( \mA \) is irreducible and satisfies the conditions in (\ref{eq:filtering conditions}). If there exists an edge index \( \overline{k} \) such that \( \mathbf{C}_{z,\overline{k}} \) in (\ref{eq_Cz}) is non-singular, and the step-sizes satisfy 
$$0 < \mu_i < \frac{2}{\rho(\mathbf{C}_{z,i})}, \quad \hbox{for all $i$,}$$ 
then \( \rho(\mathbf{B}) < 1 \). Thus, Algorithm~\ref{alg:topolms_distr} converges in the mean.

    % Assume $\mA$ be irreducible and satisfies (\ref{eq:filtering conditions}). Then, if
    %  $\mathbf{C}_{z,\overline{i}}$ is not singular for a certain $\overline{i}$, and $0 < \mu_i < \frac{2}{\rho(\mathbf{C}_{z,i})}$ for every $i$, then $\rho(\mathbf{Q})<1$, thus Algorithm \ref{alg:topolms_distr} converges in mean.
\end{Theorem}
\begin{proof}
    See Appendix C in Supplementary Material.
\end{proof}

\begin{remark}
Theorem~2 guarantees mean convergence under the assumption that only a single edge agent can asymptotically learn the true parameter \( \h^o \), provided the communication topology induces an irreducible matrix \( \mA \). This condition requires the topology to be connected. If this is not the case, it follows directly from Theorem~\ref{distr} that, to ensure convergence across the entire simplicial complex, it is sufficient for each connected component to contain at least one agent with a non-degenerate local covariance $\mathbf{C}_{z,i}$.    
\end{remark}

With regard to the mean-square behavior of Algorithm~\ref{alg:topolms_distr}, from (\ref{eq:error_recursion2}), hinging on Assumptions~3 and~4, we obtain
\begin{equation}
\mathbb{E}\{\|\widetilde{\mathbf{h}}_{\text{dec}}(n+1)\|_{\boldsymbol{\Sigma}}^2\} = \mathbb{E}\{\|\widetilde{\mathbf{h}}_{\text{dec}}(n)\|_{\boldsymbol{\Sigma}'}^2\}+\text{Tr}(\boldsymbol{\Sigma}\mathcal{A} \mathbf{M} \mathbf{G} \mathbf{M} \mathcal{A}^T)
\end{equation}
for any positive semi-definite matrix $\boldsymbol{\Sigma}$, with
\begin{align}\label{eq:rec_distr}
\boldsymbol{\Sigma}' & = \mathbb{E}\{\mathbf{B}^T(n)\boldsymbol{\Sigma}\mathbf{B}(n)\}, \\
\mathbf{G} & = \mathbb{E}\{\bg(n)\bg^T(n)\} = \text{diag}\{\sigma_{v,i}^2\mathbf{C}_{z,i}\}_{i=1,...,E}, \label{G}.
\end{align}
%\textcolor{blue}{Motivazione}  
%\textcolor{blue}{fine} 
%In the following, we will demonstrate that the method converges in mean in further general conditions, under an appropriate choice of the network's filtering coefficients $\{a_{k,l}\}$. In particular, it is required that the aggregate coefficient matrix $A$ is irreducible, meaning that it allows any two agents in the network to communicate within a finite number of steps.
%In order to choose filtering coefficients $a_{k,l}$ that compose an irreducible matrix $A$, it is necessary, according to the conditions (\ref{eq:filtering conditions}), 
%that the simplicial complex be connected with respect to the lower topology, meaning that its 1-dimensional graph substructure is connected. If this does not hold, it follows immediately from \textbf{Theorem \ref{distr}} that, to achieve convergence across the entire simplicial complex, the presence of at least one non-degenerate local covariance matrix $\mathbf{C}_{X,k}$ in each connected component of the complex is sufficient. 
%and $\|\cdot \|_{\boldsymbol{\Sigma}}$ the weighted norm. 
%Note that in ($\ref{G}$)  we used the spatially independence assumption of the zero-mean noise $\bv(n)$.
Let $\boldsymbol{\sigma} = \text{vec}(\boldsymbol{\Sigma})$ represent the vector that aggregates the columns of $\boldsymbol{\Sigma}$ vertically. Similarly, let $\boldsymbol{\sigma}' = \text{vec}(\boldsymbol{\Sigma}')$. The relationship between $\boldsymbol{\sigma}'$ and $\boldsymbol{\sigma}$ can be expressed as $
\boldsymbol{\sigma}' = \mathbf{F}\boldsymbol{\sigma},$
where
$
\mathbf{F} = \mathbb{E}\{\mathbf{B}^T(n) \otimes \mathbf{B}^T(n)\}.$
Since the the fourth-order moments of the flow signals are not known beforehand, the explicit evaluation of $\mathbf{F}$ is unavailable in this analysis. As discussed in Sec. III, the approximation $\mathbf{F} \approx \mathbf{B}^T \otimes \mathbf{B}^T$ is valid for sufficiently small step-sizes, where terms involving higher-order powers can be neglected \cite{sayed2011adaptive, sayed2014adaptation}. Under this condition, the stability of $\mathbf{F}$ holds if $\rho(\mathbf{B}) < 1$ (which is guaranteed under the assumptions of Theorem \ref{distr}). 

Using the notation $\|\cdot \|_{\boldsymbol{\sigma}}$ to indicate the norm $\|\cdot \|_{\boldsymbol{\Sigma}}$, the recursion (\ref{eq:rec_distr}) can therefore be rewritten as follows:
\begin{equation}\label{eq_var_rec_dist}
\mathbb{E}\{\|\widetilde{\mathbf{h}}_{\text{dec}}(n+1)\|_{\boldsymbol{\sigma}}^2\} = \mathbb{E}\{\|\widetilde{\mathbf{h}}_{\text{dec}}(n)\|_{\boldsymbol{\mathbf{F}\sigma}}^2\}+ \text{vec}\left(\mathcal{A} \mathbf{M} \mathbf{G} \mathbf{M} \mathcal{A}^T\right)^T \boldsymbol{\sigma} \ .
\end{equation}
If $\mathbf{F}$ is stable, the variance recursion (\ref{eq_var_rec_dist}) converges to:
\begin{equation}\label{eq:limit-point_distr}
\lim_{n \to \infty} \mathbb{E}\{\|\widetilde{\mathbf{h}}_{\text{dec}}(n)\|_{\boldsymbol{\sigma}}^2\} = \text{vec}\left(\mathcal{A} \mathbf{M} \mathbf{G} \mathbf{M} \mathcal{A}^T\right)^T (\mathbf{I} - \mathbf{F})^{-1}\boldsymbol{\sigma}.   
\end{equation}
From (\ref{eq:limit-point_distr}), we derive the MSD by choosing $\boldsymbol{\sigma}=\text{vec}(\mathbf{I})$.
% 2 pages

\section{Numerical Results}
In this section, we assess the performance of the proposed methods both on real and synthetic data.

\subsection{Topological LMS}
The first validation of Algorithm (\ref{alg:topolms}) is performed on synthetic data. The topological domain, i.e., the simplicial complex, is randomly constructed by generating a set of nodes, adding random edges between node pairs, and forming triangles from 3-node cliques, resulting in a 2D simplicial complex. In our case, $\mathcal{X}^2$ comprises $33$ nodes, $159$ edges, and $121$ triangles, with flow signals collected on edges ($k = 1$ in (\ref{signals})). We use the model in \eqref{eq:lin_obs} to generate edge signal observations, affected by additive white Gaussian measurement noise (AWGN) and with the filter degree $M=2$. In particular, the variance of the noise varies over the edges, being randomly selected from $\{  10^{-6},10^{-4},10^{-3},10^{-2}\}$. We test Algorithm (\ref{alg:topolms}), setting the step-size $\mu$  equal to $10^{-2}$, with the goal of estimating the true filter coefficients $\h^o$ that were chosen at random. 

\begin{figure}[t]
    \centering
    \includegraphics[width = .85\linewidth]{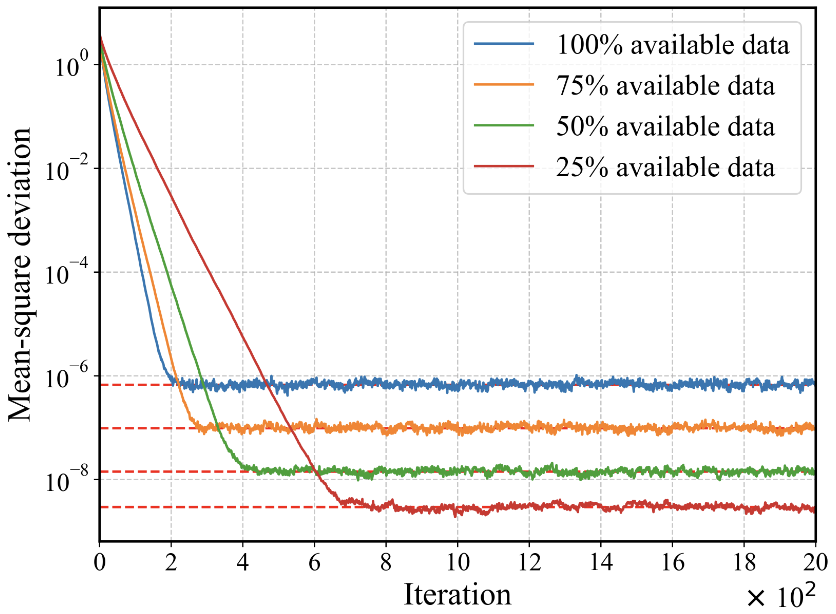}
    \caption{MSD vs iteration index, for different percentages of observed edges, compared with theoretical expressions. \vspace{-.5cm}}
    \label{fig:synth}
\end{figure}

We perform four independent simulations, with different percentages of edges available to observe the signals. Precisely, in each simulation, a prescribed percentage of edges has a non-zero probability of being sampled, that correspond to the diagonal entries of $\mathbf{P}$ that are greater than zero. In Fig.~\ref{fig:synth}, we show the  temporal MSD behavior for each case, averaged over 30 realizations, along with the corresponding theoretical steady-state values in (\ref{eq:limit-point})) depicted as dashed horizontal lines. As observed, the simulation results closely align with the theoretical predictions. We designed the test such that, for each sampling percentage $X$, the $X$-th percentile of edges with the lowest noise variance are assigned a non-zero sampling probability. In this way, from Fig.\ref{fig:synth}, we can appreciate the interesting tradeoff between convergence rate and MSD at steady state.  In particular, note that, using only the $25\%$ of edges, the convergence speed decreases, due to the slower diffusion over the complex. On the other hand, in this case, the MSD at steady state is the lowest among all simulations, since only the cleanest signal components are sampled. Vice versa, in the $100\%$ case, the algorithm has access to the entire signal for learning, but it also includes more noise: this leads to a faster convergence but a worse MSD. 

%\\
%To conclude, this test supports the theoretical convergence results of Topo-LMS in its centralized version and highlights the dependency of the algorithm's MSD and convergence rate on the sampling method.\\

\begin{figure}[t]
    \centering
    \includegraphics[width =.85\linewidth]{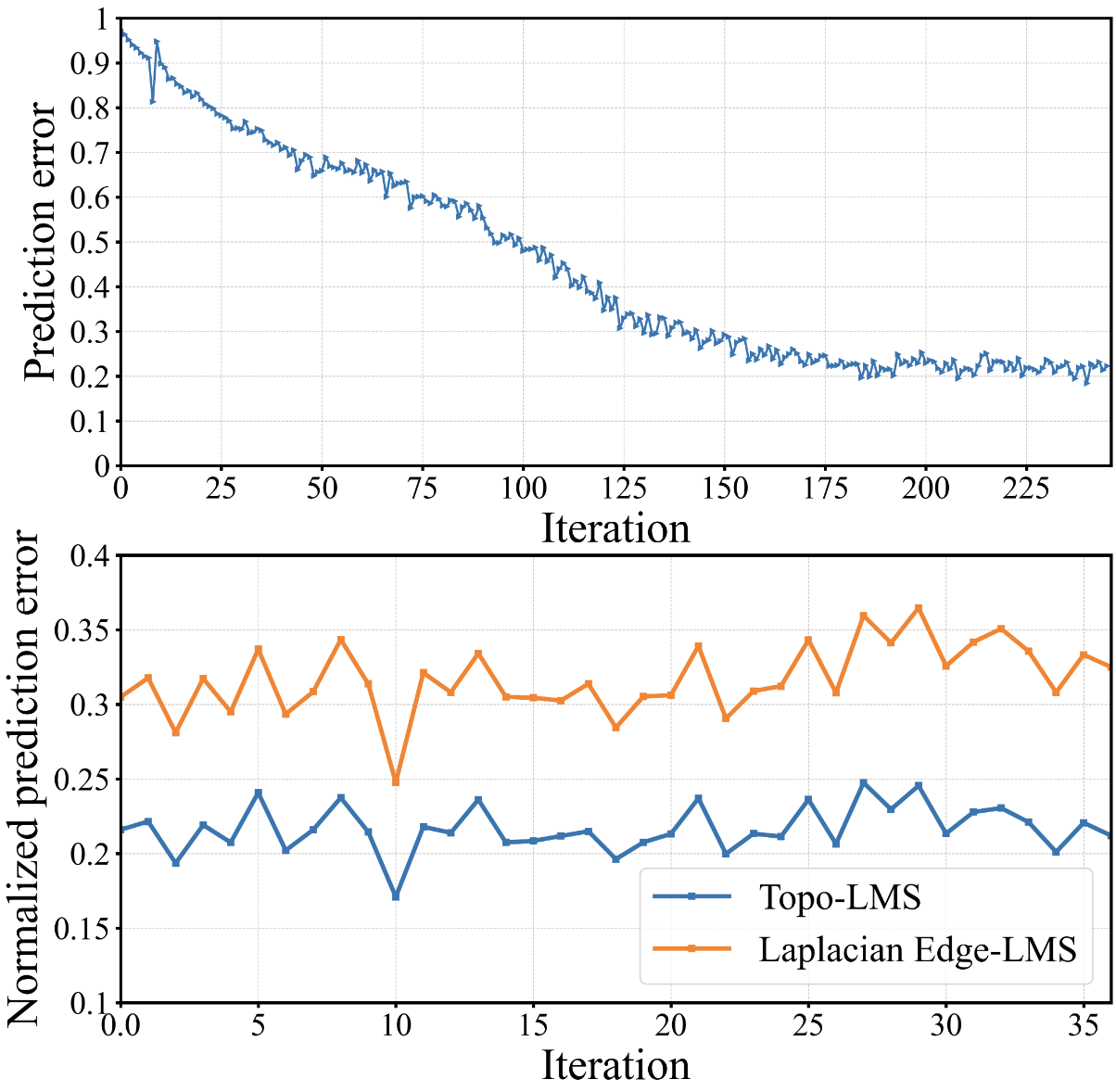}
    \caption{(\emph{Top}) Prediction error over training data. (\emph{Bottom}) Performance comparison of different methods over test data.} 
    \label{fig:real1}
\end{figure}

% We now assess the performance of the proposed algorithm on real data, considering the German National Research and Education Network operated by the German DFN-Verein (DFN), which is the communications infrastructure for Germany's broader scientific community \cite{orlowski2010sndlib}.  The network is represented as a 2-dimensional simplicial complex, enriching its graph structure with triangles placed on all its triadic cliques. The resulting simplicial complex includes $17$ nodes, $26$ edges and $5$ triangles. The dataset consists of data traffic measurements taken on February 1st, 2005 and collected on each link of the network. The measurements are aggregated and averaged over 5-minute intervals and they are expressed in Mbit/sec. We then have a time series made of 288 edge signals, the first 250 of whom are used to train the model and the last 38 for evaluating the performance. In this experiment, we assume the signal to be available all over the network, i.e.,  $\mathbf{P}=\mathbf{I}$, employing the signal model in (\ref{eq:lin_obs}) in an autoregressive fashion.  
We evaluate the proposed algorithm on real data from the German National Research and Education Network (DFN) \cite{orlowski2010sndlib}, modeled as a 2-dimensional simplicial complex by adding triangles to all triadic cliques. The resulting complex has 17 nodes, 26 edges, and 5 triangles. The dataset includes traffic measurements (in Mbit/sec) collected on February 1st, 2005, averaged over 5-minute intervals, yielding 288 edge signal snapshots. The first 250 are used for training, and the remaining 38 for testing. We assume full signal observability across the network ($\mathbf{P} = \mathbf{I}$) and use the model in (\ref{eq:lin_obs}) in an autoregressive manner. Consequently, the objective of Algorithm (\ref{alg:topolms}) is to minimize, with respect $\h$, the empirical cost function
\begin{align}
 \tilde{J}(\h)=\frac{1}{250-M}\sum_{n=M+1}^{250}\Bigg\|\mathbf{x}(n)&-\sum_{m=1}^{M}h_{m,u}\big(\L_u\big)^m \mathbf{x}(n - m) \nonumber\\
       &\hspace{-.8cm}- \sum_{m=1}^{M}h_{m,d}\big(\L_d\big)^m \mathbf{x}(n - m)\Bigg\|^2 , \nonumber
\end{align}
given the edge signal $\mathbf{x}(n)$ at time $n$ and the edge signals at previous times $\mathbf{x}(n-m)$, for $m=1, \ldots, M$, $n=M+1,...,250$. For this test, we set the filter order $M$ to $3$ and the step-size $\mu$ to $10^{-4}$. The performance of the Topo-LMS are evaluated in comparison with an alternative LMS algorithm, which only considers the graph structure of the network. This algorithm results from using a signal model that only leverages the Edge Laplacian \cite{barbarossa2020topological}, corresponding to (\ref{eq:lin_obs}) in which the coefficients $h^o_{m,u}$ are set equal to zero. The Edge Laplacian parametrization offers a fair baseline, as it constitutes the most straightforward extension of the graph-based approach in \cite{nassif2018distributed} to simplicial complexes, without incorporating Hodge theory. The results are presented in Fig.~\ref{fig:real1}, where the top panel shows the evolution of the estimation error during training, and the bottom panel illustrates the normalized prediction error using the filter coefficients learned in the first phase. Notably, Fig.~\ref{fig:real1} (Bottom) highlights that our Topo-LMS outperforms the graph-based approach in predicting the signal's evolution, thanks to its ability to leverage higher-order topological information.

\begin{figure}[t]
    \centering
    \includegraphics[width =.85\linewidth]{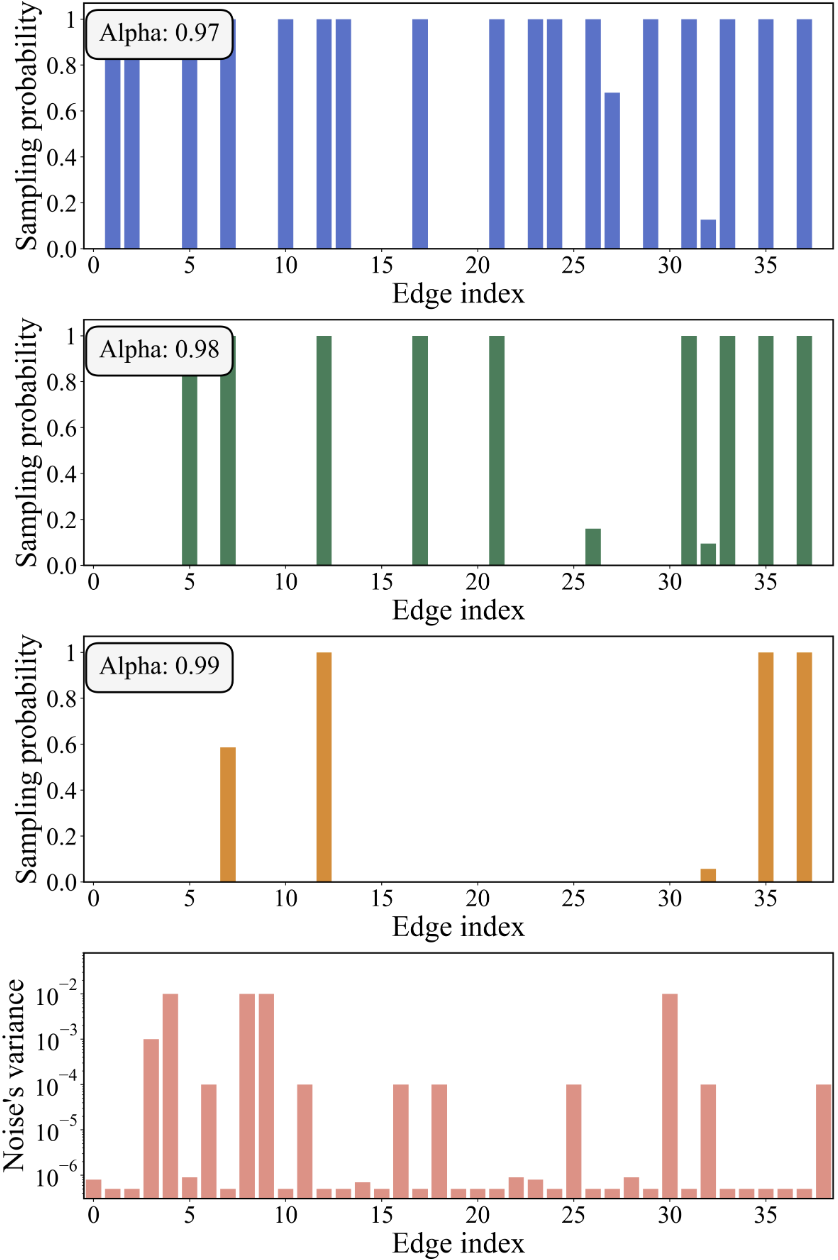}
    \caption{Optimal sampling probabilities, obtained as solution of Problem \eqref{prob_2} for different values of $\alpha$. } 
    \label{fig:sampling_2}
\end{figure}

\subsection{Optimal sampling strategies}
In this subsection, we validate the sampling strategies proposed in Sec.~III\textit{B} through numerical tests on a randomly generated 2-dimensional simplicial complex (26 nodes, 39 edges, 2 triangles), constructed using the iterative procedure from Sec.~V\textit{A}. Edge signals follow the model in (\ref{eq:lin_obs}), with filter order $M = 3$ and zero-mean AWGN whose variance varies per edge, randomly sampled from $[10^{-7}, 10^{-2}]$. We solve three sampling problems as in (\ref{prob_2}) using CVX in Python~\cite{CVX}, all under a steady-state MSD constraint $\gamma = 10^{-7}$, but with different convergence rate bounds: $\alpha_1 = 0.97$, $\alpha_2 = 0.98$, and $\alpha_3 = 0.99$. The remaining parameters are $\mu = 10^{-2}$ and $\mathbf{p}_\text{max} = \mathbf{1}$. Each solution yields a sampling probability vector that minimizes the rate while satisfying the constraints, accounting for both topological structure and noise levels.

% In this subsection, we perform numerical tests to validate the sampling strategies proposed in Sec. III.\textit{B}. A random simplicial complex of dimension $2$ is generated, with the same iterative procedure of Sec. V.\textit{A}, consisting of $26$ nodes, $39$ edges and $2$ triangles. We consider edge signals that evolve according to the model in (\ref{eq:lin_obs}), assuming zero-mean AWGN as the measurement noise, with the filter order set to $M = 3$. The noise's variance differs on each edge and it is randomly sampled from the interval $[10^{-7},10^{-2}]$. We set up three distinct sampling problems as in $(\ref{prob_2})$, using the convex optimization module CVX in Python \cite{CVX} to find the solution. For all three optimization problems, the same upper limitation $\gamma=10^{-7}$ was fixed for the MSD at steady state, whereas the convergence rate upper bounds of the algorithm are respectively $\alpha_1=0.97,\ \alpha_2=0.98$ and $\alpha_3=0.99$. The other parameters of the problems are $\mu=10^{-2}$ and $\mathbf{p}_\text{max}=\mathbf{1}$. Each solution of $(\ref{prob_2})$ is the vector of parameters that satisfies the constraints while minimizing the sampling rate. In particular, the mechanism assigns sampling probabilities taking in to account the topological information and the noise condition on each edge. 

\begin{figure}[t]
    \centering
    \includegraphics[width =.85\linewidth]{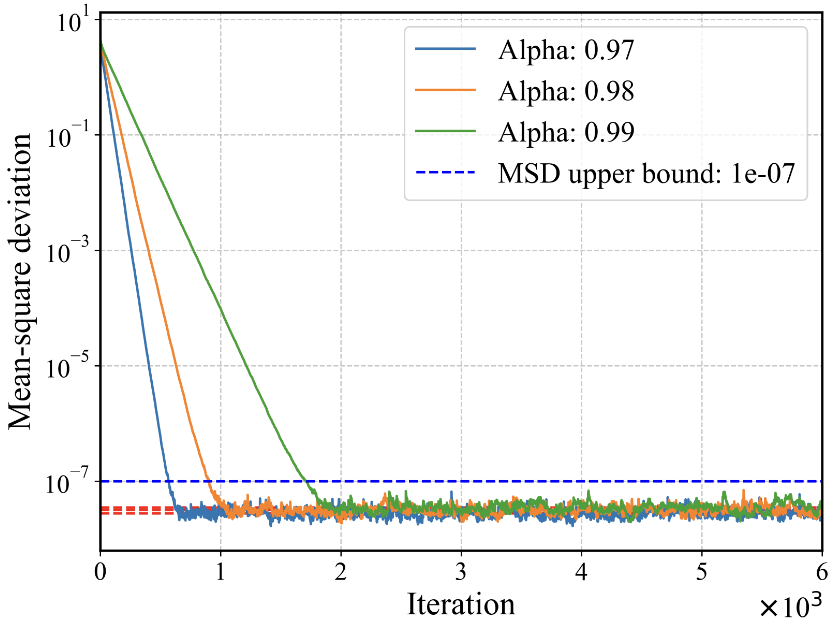}
    \caption{MSD's evolution over time with varying convergence speed of Algorithm (\ref{alg:topolms}).} 
    \label{fig:sampling_1}
\end{figure}
As illustrated in Fig.\ref{fig:sampling_2}, the expected sampling set excludes the edges on which the signal is more noisy, i.e. the measurement noises have the highest variances. In addition, it is notable that when the required convergence speed increases, i.e. the convergence rate $\alpha$ is smaller, the sampling set is enlarged, since the algorithm needs more information over the edges to speed up the process. It is also interesting to notice that the solutions are naturally sparse, as shown in Figure \ref{fig:sampling_2}, since the sampling probability are different from zero only in a small subset of edges. Finally, we generate three edge signal time series using the model in (\ref{eq:lin_obs}), where the probability of sampling at each edge is determined by the three optimal solutions obtained from the previous problems. In other words, each time series corresponds to a distinct sampling scheme. Fig.~\ref{fig:sampling_1} illustrates the different rates at which the adaptive method reaches steady state, depending on the imposed constraints on the convergence rate of the sampling mechanisms. We also observe that the theoretical MSD upper bound, $\gamma = 10^{-7}$, is clearly achieved in all three cases.

\begin{figure}[t]
    \centering
    \includegraphics[width =.85\linewidth]{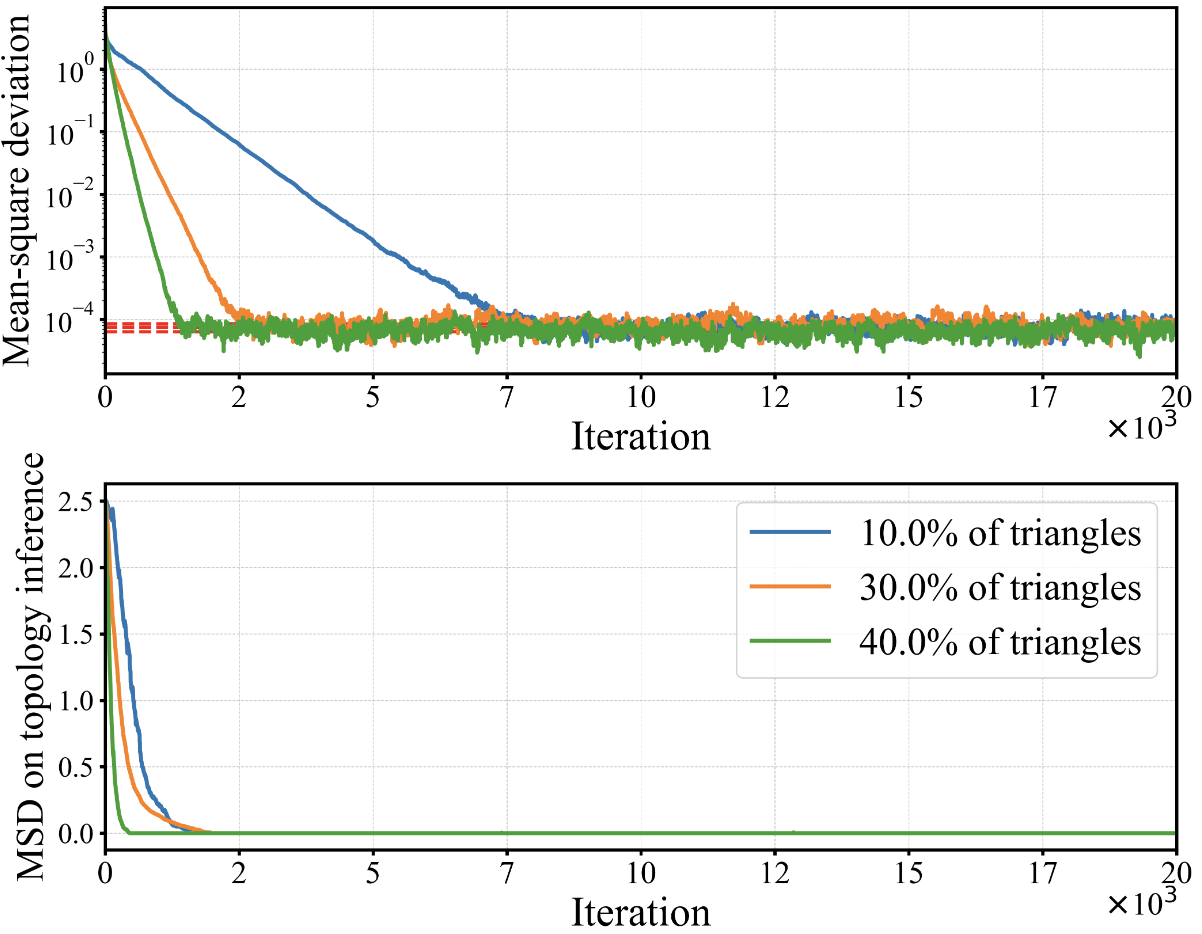}
    \caption{(\emph{Top}) Filter estimation's error over the three simulations. (\emph{Bottom}) Topology estimation's error over the three simulations. } 
    \label{fig:top_inference}
\end{figure}

\subsection{Topology Inference and Tracking}

In this subsection, we evaluate the performance of Algorithm~(\ref{eq:proxgd}) in recovering the underlying topological structure of the data. In the first experiment, we consider three randomly generated simplicial complexes of dimension $2$, which  are assumed to share the same $1$-skeleton, consisting in $35$ vertices, $115$ edges and $90$ triadic cliques. To complete each complex, a prescribed percentage of these cliques is sampled to be filled with triangles. Time series of edge signals are then generated over each complex independently, according to the generative model in (\ref{eq:lin_obs}), with AWGN measurement noise and filter order $M=2$. The noise's variance is allowed to vary over the complex, and for each edge is sampled independently from $\{10^{-6},10^{-4},10^{-3},10^{-2}\}$. We conduct three simulations, one for each simplicial complex, with the aim of recovering both the true filter coefficients $\h^o$ and the topology of the complexes, by estimating the triangles' identifiers $\mathbf{t}^o$. We then run Algorithm~(\ref{eq:proxgd}) with the step-sizes $\mu_1$, $\mu_2$ set equal to $10^{-2}$ and the thresholds $\lambda_0,\lambda_1$ set equal to $0.1$. In Fig.\ref{fig:top_inference}, we show the MSD behavior of the algorithm on the estimation of the filter coefficients (\textit{Top}) and the topology identifiers (\textit{Bottom}) in the three simulations. Notably, the learning speed is slower in the simplicial complex with $10 \%$ of the possible triangles, due to the increased complexity in correctly determining a topology with a sparse structure. However, in all three cases, the complex topology is perfectly identified, since the algorithm infers precisely the true vector $\mathbf{t}^o$, which was used to generate the data. Finally, once the topology becomes fully known, the steady-state MSD in estimating the true filter coefficients $\h^o$ reaches the theoretical limit in (\ref{eq:limit-point}), as shown in Fig.~\ref{fig:top_inference} (\textit{Top}).

\begin{figure}[t]
    \centering
    \includegraphics[width =.85\linewidth]{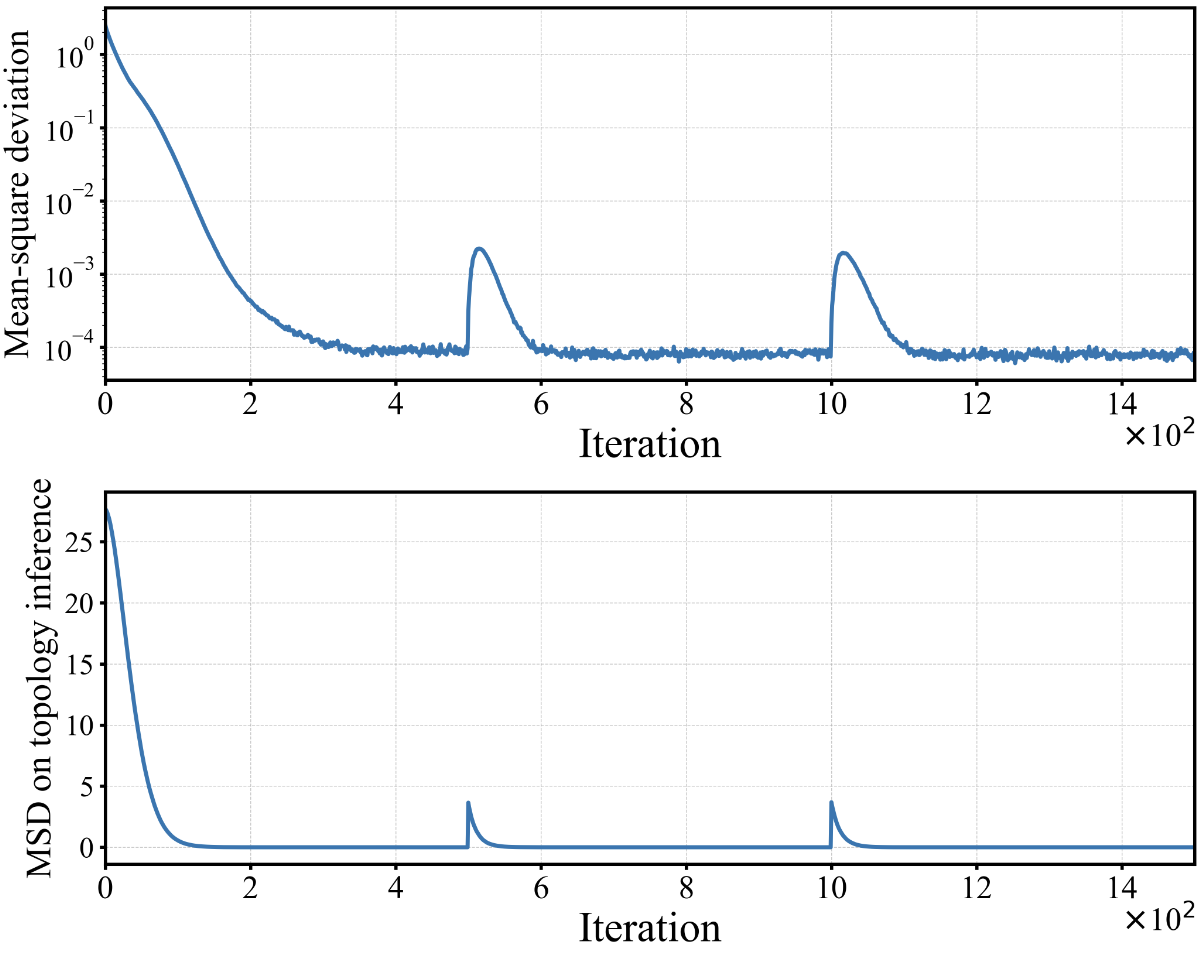}
    \caption{(\emph{Top}) Filter estimation's error with time-varying topology of the simplicial complex. (\emph{Bottom}) Topology estimation's error with time-varying topology of the simplicial complex. } 
    \label{fig:top_inference_2}
\end{figure}

We conduct an additional experiment to evaluate the algorithm’s ability to adapt to temporal changes in the topology of a simplicial complex. Specifically, we randomly generate a 2-dimensional simplicial complex consisting of 30 vertices, 139 edges, and 51 triangles. Edge signal observations are produced as in the previous simulation, using the same parameter settings. We then apply Algorithm~(\ref{eq:proxgd}) to jointly recover and track both the true filter coefficients $\h^o$ and the triangle identifiers $\mathbf{t}^o$, which encode the evolving topology of the simplicial complex.
%All other parameters are assigned values identical to those used in the prior simulation.
% according to the model in \eqref{eq:lin_obs}, where each signal is corrupted by additive white Gaussian noise (AWGN). The filter degree is set to $M=2$. For each edge, the noise variance is randomly chosen from the set $\{10^{-6},10^{-4},10^{-3},10^{-2}\}$. 
% We run Algorithm~(\ref{eq:proxgd}) to recover both the true filter coefficients $\h^o$ and the triangles' identifiers $\mathbf{t}^o$, that encode the topology of the simplicial complex. 
%In this experiment, we set the step-sizes $\mu_1$, $\mu_2$ equal to $10^{-2}$ and $10^{-1}$, respectively and the thresholds $\lambda_0,\lambda_1$ equal to $0.01$. 
To simulate a time-varying topology, we remove 4 triangles from the complex after the first 1500 signal observations, and another 4 triangles after the next 1500 observations. Figure~\ref{fig:top_inference_2} shows the algorithm’s mean squared deviation (MSD) performance in estimating the filter coefficients (top panel) and the topology identifiers (bottom panel) over time, averaged over $30$ realizations. Notably, in Figure~\ref{fig:top_inference_2} (\textit{Bottom}), the algorithm demonstrates rapid adaptation to the topology changes, accurately identifying the active triangles in the simplicial complex within a few iterations. Furthermore, as shown in Figure~\ref{fig:top_inference_2} (\textit{Top}), once the topology is correctly inferred, the algorithm quickly stabilizes and reaches a steady state after each perturbation.

\subsection{Distributed Topological LMS}

We first assess the performance of Algorithm (\ref{alg:topolms_distr}) on synthetic data. First, we randomly generate a simplicial complex, using the iterative method described in Sec. VI.\textit{A}. The resulting complex $\mathcal{X}^2$ includes $11$ nodes, $15$ edges and $10$ triangles and we process edge flow signals defined on it. We generate a dataset using the model in \eqref{eq:lin_obs}, with the filter order $M$ set equal to 2. The communication topology is connected, and the combining coefficients $a_{k,l}$ are chosen unsing a uniform rule in order to satisfy the hypothesis of Theorem \ref{distr} and ensure convergence of the method. The measurement noise is a zero-mean AWGN, where the variance on each edge is independently selected from the interval $[ 10^{-7},10^{-5} ]$. We run Algorithm (\ref{alg:topolms_distr}) aiming to recover on each edge the true filter coefficients $\h^o$ we used to generate the dataset. The step-size $\mu_k$ is set equal to $10^{-2}$ for each edge $k$. In Figure \ref{fig:LMS_distr}, we show the evolution over time of the estimation error, averaging over $30$ realizations and we also report the corresponding theoretical limit value from (\ref{eq:limit-point_distr}), which is clearly attained.

\begin{figure}[t]
    \centering
    \includegraphics[width =.85\linewidth]{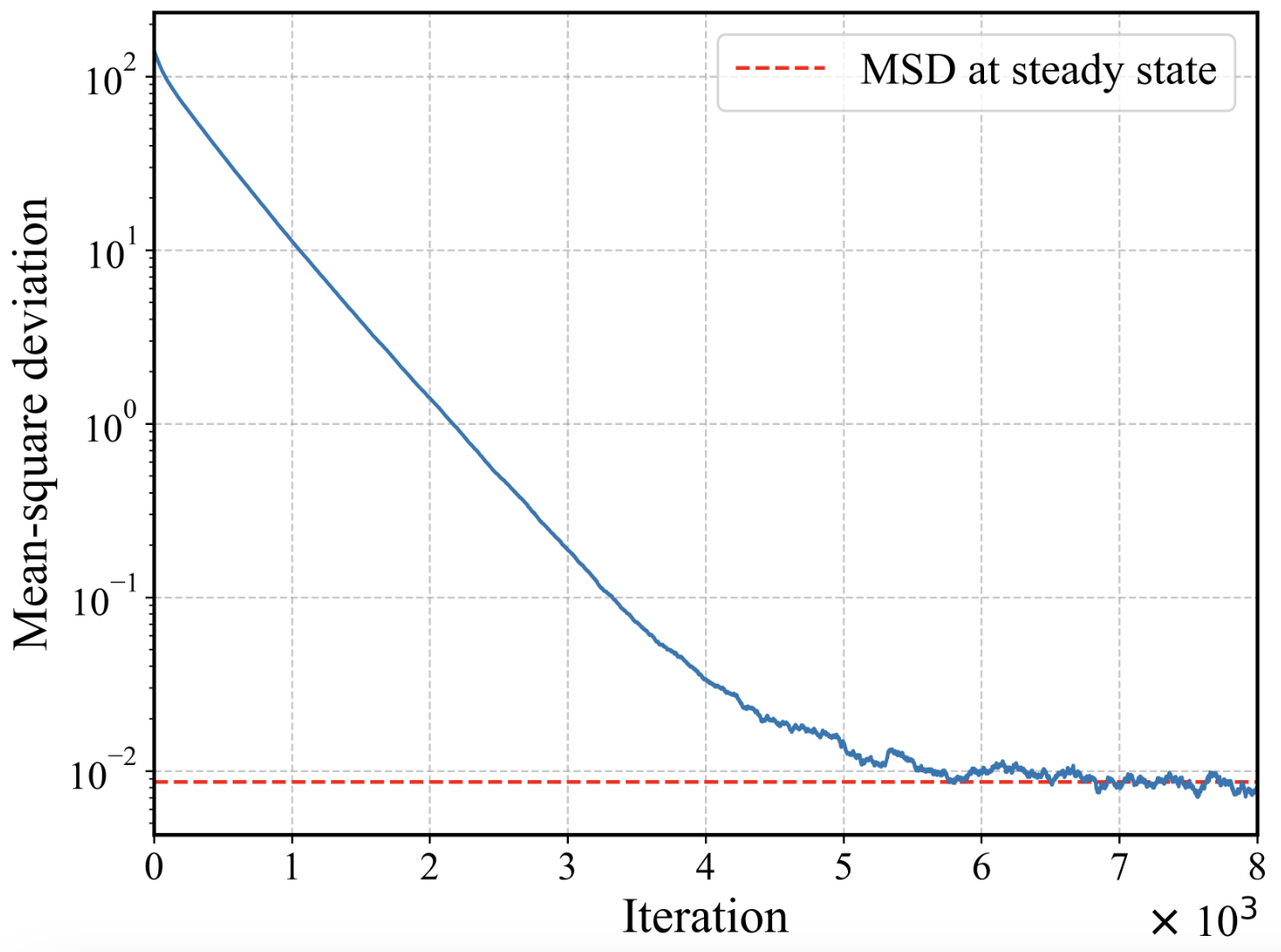}
    \caption{MSD's evolution over time of Algorithm (\ref{alg:topolms_distr}).} 
    \label{fig:LMS_distr}
\end{figure}
\begin{figure}[t]
    \centering
    \includegraphics[width =.85\linewidth]{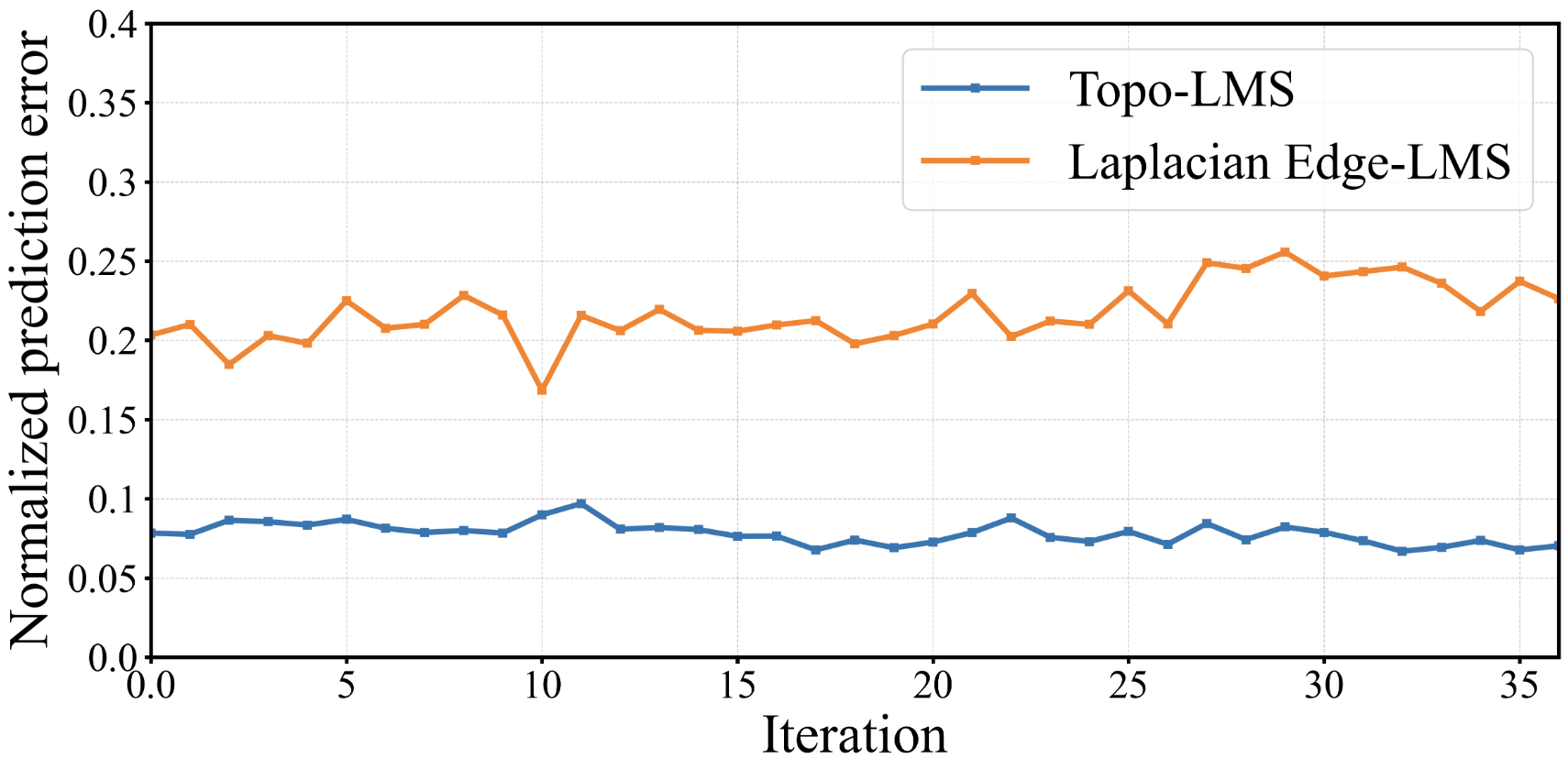}
    \caption{ Performance comparison over test data of the Topo-LMS algorithm and the graph-based Laplacian Edge-LMS algorithm.} 
    \label{fig:real_data_distr}
\end{figure}

We now test the distributed algorithm's performance on a real dataset. As for the centralized algorithm, we consider the German DFN-Verein network presented in Sec. V.\textit{A}, to process the same data in a distributed fashion. The step-size $\mu$ is chosen to be constant over the edges and equal to $10^{-1}$ and we set the filter degree $M=2$. For this experiment as well, we do not sample the signal, assuming it to be observable all over the network at each time $n$, thus imposing $\mathbf{P}=\mathbf{I}$. The generating process of the time series is assumed to be an autoregressive version of the signal model in (\ref{eq:lin_obs}). We use the first 250 observations to learn the filter coefficients, whereas the last 38 are used to test the performance of the algorithm. Compared to the centralized cases, the learning process is generally slower in distributed methods, as it is carried out by a subset of agents over the edges, that diffuse their knowledge throughout the network. Consequently, running Algorithm (\ref{alg:topolms_distr}) for a single epoch may not be sufficient to effectively train the model. Thus, since the available real data is limited, we employ the so-called \textit{empirical} stochastic gradient descent implementation, in which we run Algorithm (\ref{alg:topolms_distr}) for $N_e=30$ epochs. Formally, we extend the original dataset by defining $\mathbf{x}(i) = \mathbf{x}(d)$ for $i=1,...,N_e\cdot250$, where $d = i\ \text{mod} \ 250$ (mod being the modulo operator) and $d=1,...,250$ \cite{nassif2018distributed}. Once trained, we compare the test estimation performance with the distributed implementation of the Edge Laplacian parametrization, described in Sec. V.\textit{A}, which only leverages the graph structure of the network. As shown in Figure \ref{fig:real_data_distr}, our distributed Topo-LMS is able to reduce the prediction error w.r.t. the graph-based method, since it exploits higher-order topological information. Finally, it is noteworthy that the distributed version outperforms the centralized algorithm in (\ref{alg:topolms}) in predicting the evolution of the edge signal on test data. As shown in Fig.~\ref{fig:real1} (\textit{Bottom}) and Fig.~\ref{fig:real_data_distr}, the average prediction error of the distributed implementation is significantly lower than that of the centralized counterpart. This improvement can be attributed to the data augmentation procedure used during training, which enables the distributed algorithm to learn more effectively the model's optimal parameters.

\section{Conclusions}

In this paper, we introduced novel adaptive LMS methods for processing edge signals on simplicial complexes, leveraging their topological structure for online learning. We established convergence guarantees for the Topo-LMS algorithm, including explicit expressions for its convergence rate and steady-state MSD. A probabilistic framework was also proposed to optimize edge sampling under performance constraints. To infer the complex's topology from edge observations, we developed an alternating adaptive method that recovers latent higher-order relationships, integrated with Topo-LMS. For decentralized scenarios, we introduced a diffusion-based distributed variant with stability guarantees and closed-form MSD analysis incorporating topology, sampling, and signal features. Numerical experiments confirmed the advantages of our methods over graph-based baselines. More broadly, this work lays the groundwork for adaptive signal processing on dynamic, higher-order topological domains. While focused on edge signals, our framework extends naturally to arbitrary-order signals and more general topologies like regular cell complexes. Future directions include designing filters for other topological features and jointly processing multi-order signals via tools like the Dirac operator. Overall, our adaptive topological framework is theoretically grounded and practically flexible, with broad relevance to networked systems.

\bibliographystyle{IEEEbib}
\bibliography{biblio}

\begin{thebibliography}{10}

\bibitem{marinucci2024topological}
L.~Marinucci, C.~Battiloro, and P.~Di~Lorenzo,
\newblock ``Topological adaptive learning over cell complexes,''
\newblock in {\em 2024 32nd European Signal Processing Conference (EUSIPCO)}. IEEE, 2024, pp. 832--836.

\bibitem{battiloro_tesi}
C.~Battiloro,
\newblock ``Signal processing and learning over topological spaces,''
\newblock {\em EURASIP Library of PHD Theses}, 2024.

\bibitem{lambiotte2019networks}
R.~Lambiotte, M.~Rosvall, and I.~Scholtes,
\newblock ``From networks to optimal higher-order models of complex systems,''
\newblock {\em Nature physics}, vol. 15, no. 4, pp. 313--320, 2019.

\bibitem{Huang}
W.~Huang, T.~A.~W. Bolton, J.~D. Medaglia, D.~S. Bassett, A.~Ribeiro, and D.~Van De~Ville,
\newblock ``A graph signal processing perspective on functional brain imaging,''
\newblock {\em Proc. of the IEEE}, vol. 106, no. 5, pp. 868--885, 2018.

\bibitem{ortega2018graph}
A.~Ortega, P.~Frossard, J.~Kova{\v{c}}evi{\'c}, J.~MF Moura, and P.~Vandergheynst,
\newblock ``Graph signal processing: Overview, challenges, and applications,''
\newblock {\em Proceedings of the IEEE}, vol. 106, no. 5, pp. 808--828, 2018.

\bibitem{isufi2024graph}
E.~Isufi, F.~Gama, D.~I. Shuman, and S.~Segarra,
\newblock ``Graph filters for signal processing and machine learning on graphs,''
\newblock {\em IEEE Transactions on Signal Processing}, 2024.

\bibitem{schaub2021signal}
M.~T. Schaub, Y.~Zhu, J.B. Seby, T.~M. Roddenberry, and S.~Segarra,
\newblock ``Signal processing on higher-order networks: Livin’on the edge... and beyond,''
\newblock {\em Signal Processing}, vol. 187, pp. 108149, 2021.

\bibitem{Ghahremani}
Y.~Ghahremani and B.~Amiri,
\newblock ``A novel simplicial complex representation of social media texts: The case of twitter,''
\newblock {\em Chaos, Solitons and Fractals}, vol. 173, pp. 113642, 08 2023.

\bibitem{barbarossa2020topological}
S.~Barbarossa and S.~Sardellitti,
\newblock ``Topological signal processing over simplicial complexes,''
\newblock {\em IEEE Trans. on Signal Processing}, vol. 68, pp. 2992--3007, 2020.

\bibitem{ghorbanchian2021highord}
R.~Ghorbanchian, J.~G. Restrepo, J.~J. Torres, and G.~Bianconi,
\newblock ``Signal processing on simplicial complexes with vertex signals,''
\newblock {\em Communications Physics}, vol. 4, no. 1, pp. 120, 2021.

\bibitem{Roddenberry}
T.~M. Roddenberry, M.~T. Schaub, and M.~Hajij,
\newblock ``Signal processing on cell complexes,''
\newblock in {\em Proc. of IEEE International Conference on Acoustics, Speech and Signal Processing}, 2022, pp. 8852--8856.

\bibitem{goldberg2002combinatorial}
T.~E. Goldberg,
\newblock ``Combinatorial {L}aplacians of simplicial complexes,''
\newblock {\em Senior Thesis, Bard College}, 2002.

\bibitem{topological_slepians}
C.~Battiloro, P.~Di~Lorenzo, and S.~Barbarossa,
\newblock ``Topological slepians: Maximally localized representations of signals over simplicial complexes,''
\newblock in {\em ICASSP 2023 - 2023 IEEE International Conference on Acoustics, Speech and Signal Processing (ICASSP)}, 2023, pp. 1--5.

\bibitem{barbarossa_2}
S.~Barbarossa and S.~Sardellitti,
\newblock ``Topological signal processing: Making sense of data building on multiway relations,''
\newblock {\em IEEE Signal Processing Magazine}, vol. 37, no. 6, pp. 174--183, 2020.

\bibitem{SCF_9893391}
M.~Yang, E.~Isufi, M.~T. Schaub, and G.~Leus,
\newblock ``Simplicial convolutional filters,''
\newblock {\em IEEE Trans. on Signal Processing}, vol. 70, pp. 4633--4648, 2022.

\bibitem{grimaldi2025topologicaldictionarylearning}
Grimaldi E., Battiloro C., and Di~Lorenzo P.,
\newblock ``Topological dictionary learning,''
\newblock {\em arXiv preprint}, 2025.

\bibitem{papillon2024}
M.~Papillon, S.~Sanborn, M.~Hajij, and N.~Miolane,
\newblock ``Architectures of topological deep learning: A survey of message-passing topological neural networks,''
\newblock {\em arXiv preprint}, 2024.

\bibitem{battiloro2024generalized}
C.~Battiloro, L.~Testa, L.~Giusti, S.~Sardellitti, P.~Di~Lorenzo, and S.~Barbarossa,
\newblock ``Generalized simplicial attention neural networks,''
\newblock {\em IEEE Transactions on Signal and Information Processing over Networks}, 2024.

\bibitem{bodnar2021weisfeiler}
C.~Bodnar, F.~Frasca, Y.~Guang Wang, N.~Otter, G.~Montufar, P.~Li{\`o}, and M.~M. Bronstein,
\newblock ``Weisfeiler and {L}ehman go topological: Message passing simplicial networks,''
\newblock in {\em ICLR 2021 Workshop on Geometrical and Topological Representation Learning}, 2021.

\bibitem{SAT}
G.~Christopher, W.~Jin, C.~Bodnar, and P.~Liò,
\newblock ``Simplicial attention networks,''
\newblock {\em arXiv preprint arXiv:2204.09455}, 2022.

\bibitem{Chen_power_grid}
Y.~Chen, R.~A. Jacob, Y.~R. Gel, Jie Zhang, and H.~V. Poor,
\newblock ``Learning power grid outages with higher-order topological neural networks,''
\newblock {\em IEEE Transactions on Power Systems}, vol. 39, no. 1, pp. 720--732, 2024.

\bibitem{battiloro2023latentgraphlatenttopology}
C.~Battiloro, I.~Spinelli, L.~Telyatnikov, M.~Bronstein, S.~Scardapane, and P.~Di Lorenzo,
\newblock ``From latent graph to latent topology inference: Differentiable cell complex module,''
\newblock {\em arXiv preprint}, 2023.

\bibitem{sayed2011adaptive}
A.~H. Sayed,
\newblock {\em Adaptive filters},
\newblock John Wiley \& Sons, 2011.

\bibitem{sayed2014adaptation}
A.~H. Sayed,
\newblock ``Adaptation, learning, and optimization over networks,''
\newblock {\em Foundations and Trends{\textregistered} in Machine Learning}, vol. 7, 2014.

\bibitem{sayed2022inference}
A.~H Sayed,
\newblock {\em Inference and Learning from Data: Learning}, vol.~3,
\newblock Cambridge University Press, 2022.

\bibitem{Graph_inference}
S.~Vlaski, H.~P. Maretić, R.~Nassif, P.~Frossard, and A.~H. Sayed,
\newblock ``Online graph learning from sequential data,''
\newblock in {\em 2018 IEEE Data Science Workshop (DSW)}, 2018, pp. 190--194.

\bibitem{di2016adaptive}
P.~Di Lorenzo, S.~Barbarossa, P.~Banelli, and S.~Sardellitti,
\newblock ``Adaptive least mean squares estimation of graph signals,''
\newblock {\em IEEE Trans. on Signal and Information Proc. over Networks}, vol. 2, no. 4, pp. 555--568, 2016.

\bibitem{isufi2020observing}
E.~Isufi, P.~Banelli, P.~Di Lorenzo, and G.~Leus,
\newblock ``Observing and tracking bandlimited graph processes from sampled measurements,''
\newblock {\em Signal Processing}, vol. 177, pp. 107749, 2020.

\bibitem{di2018adaptive}
P.~Di Lorenzo, P.~Banelli, E.~Isufi, S.~Barbarossa, and G.~Leus,
\newblock ``Adaptive graph signal processing: Algorithms and optimal sampling strategies,''
\newblock {\em IEEE Trans. on Signal Processing}, vol. 66, no. 13, pp. 3584--3598, 2018.

\bibitem{di2017distributed}
P.~Di~Lorenzo, P.~Banelli, S.~Barbarossa, and S.~Sardellitti,
\newblock ``Distributed adaptive learning of graph signals,''
\newblock {\em IEEE Transactions on Signal Processing}, vol. 65, no. 16, pp. 4193--4208, 2017.

\bibitem{nassif2018distributed}
R.~Nassif, C.~Richard, J.~Chen, and A.~H. Sayed,
\newblock ``Distributed diffusion adaptation over graph signals,''
\newblock in {\em Proc. of ICASSP}, 2018, pp. 4129--4133.

\bibitem{hua2020online}
F.~Hua, R.~Nassif, C.~Richard, H.~Wang, and A.~H. Sayed,
\newblock ``Online distributed learning over graphs with multitask graph-filter models,''
\newblock {\em IEEE Trans. on Sig. and Inform. Proc. over Networks}, vol. 6, pp. 63--77, 2020.

\bibitem{isufi2025topologicalsignalprocessinglearning}
E.~Isufi, G.~Leus, B.~Beferull-Lozano, S.~Barbarossa, and P.~Di~Lorenzo,
\newblock ``Topological signal processing and learning: Recent advances and future challenges,''
\newblock {\em Signal Processing}, 2025.

\bibitem{Money}
R.~Money, J.~Krishnan, B.~Beferull-Lozano, and E.~Isufi,
\newblock ``Online edge flow imputation on networks,''
\newblock {\em IEEE Signal Processing Letters}, vol. 30, pp. 115--119, 2023.

\bibitem{Money_kalman}
R.~Money, M.~Sabbaqi, J.~Krishnan, B.~Beferull-Lozano, and E.~Isufi,
\newblock ``Kalman filtering for simplicial processes,''
\newblock in {\em 2024 58th Asilomar Conference on Signals, Systems, and Computers}, 2024, pp. 49--53.

\bibitem{yang2023online}
M.~Yang, B.~Das, and E.~Isufi,
\newblock ``Online edge flow prediction over expanding simplicial complexes,''
\newblock in {\em Proc. of ICASSP}, 2023, pp. 1--5.

\bibitem{Yan2025}
Y.~Yan, T.~Xie, and E.~E. Kuruoglu,
\newblock ``Adaptive joint estimation of temporal vertex and edge signals,''
\newblock {\em IEEE Transactions on Signal and Information Processing over Networks}, pp. 1--15, 2025.

\bibitem{Krishnan2024}
J.~Krishnan, R.~Money, B.~Beferull-Lozano, and E.~Isufi,
\newblock ``Simplicial vector autoregressive models,''
\newblock {\em IEEE Transactions on Signal Processing}, vol. 72, pp. 5454--5469, 2024.

\bibitem{Hatcher}
A.~Hatcher,
\newblock {\em Algebraic topology},
\newblock Cambridge University Press, 2002.

\bibitem{lim2020hodge}
L.~Lim,
\newblock ``Hodge {L}aplacians on graphs,''
\newblock {\em Siam Review}, vol. 62, no. 3, pp. 685--715, 2020.

\bibitem{yang2023convolutional}
M.~Yang and E.~Isufi,
\newblock ``Convolutional learning on simplicial complexes,''
\newblock {\em arXiv preprint arXiv:2301.11163}, 2023.

\bibitem{Sheng-DeWang}
W.~Sheng-De, K.~Te-Son, and H.~Chen-Fa,
\newblock ``Trace bounds on the solution of the algebraic matrix riccati and lyapunov equation,''
\newblock {\em IEEE Transactions on Automatic Control}, vol. 31, no. 7, pp. 654--656, 1986.

\bibitem{Avriel}
M.~Avriel, W.~E. Diewert, S.~Schaible, and I.~Zang,
\newblock {\em Generalized Concavity},
\newblock Society for Industrial and Applied Mathematics, 2010.

\bibitem{CVX}
S.~Boyd and L.~Vandenberghe,
\newblock {\em Convex optimization},
\newblock Cambridge university press, 2004.

\bibitem{di2020distributed}
P.~Di~Lorenzo, S.~Barbarossa, and S.~Sardellitti,
\newblock ``Distributed signal processing and optimization based on in-network subspace projections,''
\newblock {\em IEEE Trans. on Signal Processing}, vol. 68, pp. 2061--2076, 2020.

\bibitem{cattivelli2009diffusion}
F.~S Cattivelli and A.~H Sayed,
\newblock ``Diffusion {LMS} strategies for distributed estimation,''
\newblock {\em IEEE Trans. on Signal Processing}, vol. 58, no. 3, pp. 1035--1048, 2009.

\bibitem{norris1998markov}
J.~R. Norris,
\newblock {\em Markov Chains},
\newblock Cambridge Series in Statistical and Probabilistic Mathematics. Cambridge University Press, 1998.

\bibitem{orlowski2010sndlib}
S.~Orlowski, R.~Wess{\"a}ly, M.~Pi{\'o}ro, and A.~Tomaszewski,
\newblock ``Sndlib 1.0—survivable network design library,''
\newblock {\em Networks: An International Journal}, vol. 55, no. 3, pp. 276--286, 2010.

\end{thebibliography}

\clearpage

\twocolumn[{%
    \begin{center}
        \textbf{\Large Supplementary Material} \medskip
    \end{center}
}]

\appendices
\section{Proof Of Theorem 1}
We begin by establishing a preliminary result.
\begin{lemma}\label{lem}
    Let $\mathbf{Z}$ be a positive semi-definite matrix and let $\lambda_{\text{max}}(\mathbf{Z})$ be its maximal eigenvalue. Given $\delta >0$, for any $0< \mu < \displaystyle\frac{1}{\lambda_{\text{max}}(\mathbf{Z})}-\delta $ there exists a matrix $\mathbf{R}(\mu)$ such that 
\begin{equation}
(\mathbf{I}-\mu \mathbf{Z})^{-1}=\mathbf{I}+\mu\mathbf{R}(\mu), \label{lem1}
\end{equation}
with $\Vert \mathbf{R}(\mu)\Vert_{2} < 1/\delta$ (not depending on $\mu$).
\end{lemma}
\begin{proof} The proof can be derived from the matrix expansion 
\begin{eqnarray}
(\mathbf{I}-\mu \mathbf{A})^{-1} &=& \sum_{m=0}^{\infty}\mu^{m}\mathbf{Z}^{m}
  \notag \\ 
  &=& \mathbf{I} + \mu\sum_{m=1}^{\infty}\mu^{m-1}\mathbf{Z}^{m} = \mathbf{I}+\mu\mathbf{R}(\mu). \notag
\end{eqnarray}
Note that $\Vert \mathbf{R}(\mu)\Vert_{2} < 1/\delta$, concluding the proof.
\end{proof}
%By the assumption on $\mu$, we have that $\mathbf{I}-\mu \mathbf{A}$ is positive definite, with eigenvalues equal to $1-\mu \lambda_1,..., 1-\mu \lambda_n$, where $\lambda_1,...,\lambda_n$ are the eigenvalues of $\mathbf{A}$. We can then write 
%\begin{equation}
%(\mathbf{I}-\mu \mathbf{A})^{-1}=\mathbf{U}^T \text{diag}\left(  \frac{1}{1-\mu \lambda_1},...,\frac{1}{1-\mu \lambda_n}\right)\mathbf{U}, \label{lem2} 
%\end{equation}
%so that we obtain
%\begin{eqnarray}
%\mathbf{I}-(\mathbf{I}-\mu \mathbf{A})^{-1}&=&\mathbf{I}-\mathbf{U}^T \text{diag}\left(  \frac{1}{1-\mu \lambda_1},...,\frac{1}{1-\mu \lambda_n}\right)\mathbf{U} \label{lem3} \\
%&=& \mathbf{U}^T \left[\mathbf{I}-\text{diag}\left(  \frac{1}{1-\mu \lambda_1},...,\frac{1}{1-\mu \lambda_n}\right)\right]\mathbf{U} \notag \\
%&=& \mathbf{U}^T \text{diag}\left(  \frac{\mu \lambda_1}{1-\mu \lambda_1},...,\frac{\mu \lambda_n}{1-\mu \lambda_n}\right)\mathbf{U} \notag \\
%&=& \mu \mathbf{U}^T \text{diag}\left(  \frac{\lambda_1}{1-\mu \lambda_1},...,\frac{\lambda_n}{1-\mu \lambda_n}\right)\mathbf{U}=\mu \mathbf{C}({\mu}).\notag
%\end{eqnarray}

We now proceed proving Theorem 1, starting with the derivation of (\ref{eq:limit_var}). Let us recall expression (\ref{eq:limit-point}) with $\boldsymbol{\sigma}=\text{vec}(\mathbf{I})$: 
\begin{equation}\label{eq:limit-point2}
\lim_{n \to \infty} \mathbf{E}\{\|\widetilde{\mathbf{h}}(n)\|^2\} = \mu^2\text{vec}\left(\mathbf{G}\right)^T  (\mathbf{I} - \mathbf{F})^{-1}\text{vec}(\mathbf{I})\ ,   
\end{equation}
where $\mathbf{F}\cong\mathbf{Q} \otimes \mathbf{Q}$ and $\mathbf{Q}=\mathbf{I}-\mu \mathbf{C}_X$. Using the properties of the Kronecker product, the last term in (\ref{eq:limit-point2}) can be recast as:
\begin{align}
  (\mathbf{I} - &\mathbf{F})^{-1}\text{vec}\left(\mathbf{I}\right) =  \sum_{m=0}^{\infty}(\mathbf{Q} \otimes \mathbf{Q})^{m}\text{vec}\left(\mathbf{I}\right) 
  \notag \\ &=  \sum_{m=0}^{\infty}\left(\mathbf{Q}^{m} \otimes \mathbf{Q}^{m}\right)\text{vec}\left(\mathbf{I}\right) =\sum_{m=0}^{\infty}\text{vec}\left(\mathbf{Q}^{2m}\right)
  \notag \\ &= \text{vec}\left((\mathbf{I} - \mathbf{Q}^{2})^{-1}\right) . \label{equation_1}   
\end{align}
Since $\mathbf{Q}=\mathbf{I}-\mu \mathbf{C}_X$, the last matrix in (\ref{equation_1}) reads as:
\begin{align}
(\mathbf{I} - \mathbf{Q}^{2})^{-1} &= (2\mu \mathbf{C}_X - {\mu}^2 {\mathbf{C}^2_X})^{-1} = \frac{1}{2\mu} \mathbf{C}^{-1}_X \left(\mathbf{I} - \frac{\mu}{2}\mathbf{C}_X\right)^{-1}\notag \\ & \overset{(a)}{=}\, \frac{1}{2\mu} \mathbf{C}_X^{-1} \left(\mathbf{I} + \frac{\mu}{2}\mathbf{R}(\mu)\right) \label{equation_2}
\end{align}
where in $(a)$ we have applied Lemma 2 for $\mZ=\frac{1}{2}\mathbf{C}_X$. Note that the conditions of Lemma 2 hold in this case, since we have assumed that the step-size $\mu$ satisfies (\ref{eq:step}). Finally, recasting (\ref{eq:limit-point2}) through (\ref{equation_1})-(\ref{equation_2}), we get:
\begin{eqnarray}
\lim_{n \to \infty} \mathbb{E}\{\|\widetilde{\mathbf{h}}(n)\|^2\} &=& \mu^2\text{vec}\left(\mathbf{G}\right)^T\text{vec}\left((\mathbf{I} - \mathbf{Q}^{2})^{-1}\right) \notag \\
&=&\mu^2\text{Tr}\left(\mathbf{G}(\mathbf{I} - \mathbf{Q}^{2})^{-1}\right) \label{acca} \notag \\ 
&=& \frac{\mu}{2}\text{Tr}\left(\mathbf{G}\mathbf{C}_X^{-1}\right) + \frac{{\mu}^2}{4} \text{Tr}\left(\mathbf{G}\mathbf{C}_X^{-1}\mathbf{R}(\mu)\right) \notag \\ 
&=& \frac{\mu}{2}\text{Tr}\left(\mathbf{G}\mathbf{C}_X^{-1}\right) + \textit{O}({\mu}^2), \notag
\end{eqnarray}
which is the expression claimed in (\ref{eq:limit_var}). 

Now, we prove the second part of the Theorem, deriving the convergence rate $\alpha$ in (\ref{eq:convergence_rate}). We recall that $\alpha=\|\mathbf{F}\|_2$, with $\mathbf{F}\cong\mathbf{Q} \otimes \mathbf{Q}$ for a sufficiently small step-size $\mu$. Since $\mathbf{Q}=\mathbf{I}-\mu \mathbf{C}_X$, we have that:
    $\|\mathbf{F}\|_2=\|\mathbf{Q}\otimes \mathbf{Q}\|_2 =
    \|\mathbf{Q}\|^2_2 =\|\mathbf{I}-\mu \mathbf{C}_X\|_2^2$.
Moreover, it holds that:
\begin{eqnarray}
     \|\mathbf{I}-\mu \mathbf{C}_X\|_2^2&=& \text{\text{max}} \{ (1-\mu \delta)^2, (1-\mu \nu)^2\}\notag \\
    &\leq& 1-2 \mu \delta + \mu^2 \nu^2 \notag \\ &=& 1-2\mu \delta \left( 1-\frac{\mu \nu^2}{2 \delta}\right)  \notag\\
    & \overset{(a)}{\approx}& 1-2\mu \delta \notag
\end{eqnarray}
where $(a)$ holds when $ \mu \ll 2 \delta/\nu^2$.
%which gives an upper bound for the convergence rate of Algorithm (\ref{alg:topolms}). In particular, it follows: \begin{equation}\alpha \approx  1-2\mu \delta \notag \end{equation} 
%when $ \mu << 2 \delta/\nu^2$, which proves the statement. 
This proves (\ref{eq:convergence_rate}), and concludes the proof of Theorem 1. 

\section{Proof Of Lemma 1}
Since problem (\ref{eq:prox}) is separable in each coordinate $u_i$ of $\mathbf{u}$, we can solve it independently for each component. Thus, we focus on the scalar case:
\begin{align}\label{eq:Prox_proof}
\mathcal{H}_{\lambda_0}^{\lambda_1}(v_i) =
\underset{u_i \in [0,1]}{\arg\min} \Bigg\{ &\frac{1}{2} (u_i - v_i)^2 
+ \lambda_0 \mathds{1}_{\{u_i \neq 0\}} + \lambda_1 \mathds{1}_{\{u_i \neq 1\}} \Bigg\} 
\end{align}
for $i=1,\ldots,E$, where \( \mathds{1}_{\{A\}} \) is the indicator function that equals 1 if condition \( A \) holds true, and 0 otherwise. It is straightforward that $\mathcal{H}_{\lambda_0}^{\lambda_1}(v_i)$ can only assume values in $\{0, v_i,1\}$. Evaluating the objective function $F$ of (\ref{eq:Prox_proof}) at \( u_i = \{0, v_i, 1\} \) with $v_i \in (0,1)$, we obtain:
\[
F(0) = \frac{1}{2} v_i^2 + \lambda_1, \;\;
F(v_i) = \lambda_0 + \lambda_1, \;\;
F(1) = \frac{1}{2} (1 - v_i)^2 + \lambda_0.
\]
Comparing \( F(0) \leq F(v_i) \) leads to \( v_i^2 \leq 2\lambda_0 \), i.e., \( v_i \leq \sqrt{2\lambda_1} \).  
Comparing \( F(v_i) \leq F(1) \) leads to \( (1 - v_i)^2 \geq 2\lambda_1 \), i.e., \( v_i \leq 1 - \sqrt{2\lambda_1} \). Thus, the optimal solution of (\ref{eq:Prox_proof}) is:
\[
\mathcal{H}_{\lambda_0}^{\lambda_1}(v_i) =
\begin{cases} 
0, & v_i \leq \sqrt{2\lambda_0}, \\
v_i, & \sqrt{2\lambda_0} < v_i < 1 - \sqrt{2\lambda_1}, \\
1, & v_i \geq 1 - \sqrt{2\lambda_1}.
\end{cases}
\] 
for $i=1,\ldots,E$, concluding the proof of the lemma.

% \section*{Proof Of Theorem 2}
% We recall that the convergence rate $\alpha=\|\mathbf{F}\|_2$, with $\mathbf{F}\cong\mathbf{Q} \otimes \mathbf{Q}$ for a sufficiently small step-size $\mu$. 
% Since $\mathbf{Q}=\mathbf{I}-\mu \mathbf{C}_X$, we have that:
% \begin{eqnarray}
%     \|\mathbf{F}\|_2&=& \|\mathbf{Q}\otimes \mathbf{Q}\|_2   \notag   \\
%     &=& \|\mathbf{Q}\|^2_2 \notag  \\
%     &=&\|\mathbf{I}-\mu \mathbf{C}_X\|_2^2. \notag 
% \end{eqnarray}
% Moreover, it holds that:
% \begin{eqnarray}
%      \|\mathbf{I}-\mu \mathbf{C}_X\|_2^2&=& \text{\text{max}} \{ (1-\mu \delta)^2, (1-\mu \nu)^2\}\notag \\
%     &\leq& 1-2 \mu \delta + \mu^2 \nu^2 \notag \\ &=& 1-2\mu \delta \left( 1-\frac{\mu \nu^2}{2 \delta}\right),  \notag
% \end{eqnarray}
% which gives an upper bound for the convergence rate of Algorithm (\ref{alg:topolms}). In particular, it follows: \begin{equation}\alpha \approx  1-2\mu \delta \notag \end{equation} 
% when 
% $ \mu << 2 \delta/\nu^2$, which proves the statement. 

\section{Proof Of Theorem 2}
To simplify the discussion, let us define:
\begin{equation}
    \mathbf{T}_i=I_{2M+1}-\mu_i \mathbf{C}_{z,i}, \quad
    \mathcal{T}= \text{diag}\{\mathbf{T}_i\}_{i=1,...,E} \ , \notag 
\end{equation}
thus we have that $\mathbf{B}=\mathcal{A}\mathcal{T}$. Since $0 < \mu_i < \frac{2}{\rho(\mathbf{C}_{z,i})}$ for all $i$, it follows that $\rho(\mathbf{T}_i) \leq 1$ for all $i$, and consequently $\rho(\mathcal{T}) \leq 1$. In addition, since $\mathbf{C}_{z,{\overline{k}}}$ is positive definite for a certain $\overline{k}$, we have that $\rho(\mathbf{T}_{\overline{k}})<1$.

Arguing by contradiction, suppose that $\rho(\mathbf{B}) \geq 1$, and let $\lambda$ be an eigenvalue of $\mathbf{B}$ with $|\lambda| \geq 1$, and let $\mathbf{v}$ be a corresponding eigenvector. Since $\mathbf{v},\, \mathbf{B}\mathbf{v} \in \mathbb{R}^{E \cdot (2M+1)}$, we have the following expressions:
\begin{equation}
    \mathbf{v} = \operatorname{col}\{\mathbf{v}_i\}_{i=1}^E, \qquad
    \mathbf{B}\mathbf{v} = \operatorname{col}\{[\mathbf{Q}\mathbf{v}]_i\}_{i=1}^E, \notag
\end{equation}
where $\mathbf{v}_i,\, [\mathbf{Q}\mathbf{v}]_i \in \mathbb{R}^{2M+1}$, and $\operatorname{col}(\cdot)$ denotes the column stacking operator. As a consequence of the block structure of $\mathcal{A}$ and $\mathcal{T}$ and since $\bv$ is an eigenvector, we have:
\begin{equation}
    \lambda \bv_i = [\mathbf{Q}\bv]_i=\sum_{j \in \mathcal{N}_i} a_{i,j}\mathbf{T}_j \bv_j, \quad i=1,...,E \ .
\end{equation}
% Consequently, since $|\lambda|\geq 1$, it holds:
% \begin{align}\label{inequality}
%     \| \bv_i \| &\leq \Big\|\sum_{j \in \mathcal{N}_i} a_{i,j}\mathbf{T}_j\bv_j   \Big\| \notag \\
%     &\leq \sum_{j \in \mathcal{N}_i} a_{k,j} \| \mathbf{T}_j\bv_j \| \,\leq\, \sum_{j \in \mathcal{N}_i} a_{i,j} \| \bv_j\|,
% \end{align}
% where the final inequality in (\ref{inequality}) comes from the fact that $\rho(\mathbf{T}_i)\leq1$ and $\mathbf{T}_i$ is diagonalizable for every $i$. Now, let $\bv_{h}$ be a component of $\bv$ such that $\|\bv_{h}\| \geq \|\bv_i\|$ for $i=1,...,E$. Since $\sum_{j \in \mathcal{N}_i} a_{h,j}=1$ and $a_{h,j}\geq 0$ for every $j$ (cf. \eqref{eq:filtering conditions}), from (\ref{inequality}) we get that:
% \begin{eqnarray}
%     \|\bv_i\|= \|\mathbf{T}_i \bv_i\|= \|\bv_h\| \quad \text{ for every } i \in \mathcal{N}_h.
% \end{eqnarray}
% This implies that the norm is maximal for all the neighbors of $\bv_h$. Since $\mA$ is irreducible, each pair of indices communicates and, for the same line of reasoning, it holds that $\|\bv_i\|=\|\bv_h\|$ and $\|\bv_i\|= \|\mathbf{T}_i \bv_i\|$, for every $i=1,...,E$. But this is not possible, since by hypothesis $\rho(\mathbf{T}_{\overline{k}})<1$ and then $\|\bv_{\overline{k}}\|> \|\mathbf{T}_{\overline{k}} \bv_{\overline{k}}\|$.
Taking norms on both sides and using the triangle inequality:
\begin{align}\label{inequality2}
    \|\mathbf{v}_i\| &= \frac{1}{|\lambda|} \left\| \sum_{j \in \mathcal{N}_i} a_{i,j} \mathbf{T}_j \mathbf{v}_j \right\| \leq \frac{1}{|\lambda|} \sum_{j \in \mathcal{N}_i} a_{i,j} \| \mathbf{T}_j \mathbf{v}_j \| \notag\\
    &\leq \frac{1}{|\lambda|} \sum_{j \in \mathcal{N}_i} a_{i,j} \| \mathbf{v}_j \|,
\end{align}
where the last inequality holds because each $\mathbf{T}_j$ is diagonalizable with spectral radius at most $1$.

Now, define index $h$ such that $\|\mathbf{v}_h\| = \max_{i=1,\dots,E} \|\mathbf{v}_i\|$. Applying the  inequality (\ref{inequality2}) at $i = h$ and noting that $\sum_{j \in \mathcal{N}_h} a_{h,j} = 1$ and $a_{h,j} \geq 0$, we obtain:
\begin{equation*}
    \|\mathbf{v}_h\| \leq \frac{1}{|\lambda|} \sum_{j \in \mathcal{N}_h} a_{h,j} \|\mathbf{v}_j\| \leq \frac{1}{|\lambda|} \|\mathbf{v}_h\|.
\end{equation*}
Since $|\lambda| \geq 1$, this chain of inequalities can only hold if all inequalities are equalities. Therefore, for all $j \in \mathcal{N}_h$:
\[
    \|\mathbf{v}_j\| = \|\mathbf{v}_h\| \quad \text{and} \quad \|\mathbf{T}_j \mathbf{v}_j\| = \|\mathbf{v}_j\|.
\]
This shows that all neighbors of $h$ attain the same maximum norm and are invariant under their corresponding $\mathbf{T}_j$ in norm. Since $\mathbf{A}$ is irreducible, a similar argument extends this equality of norms to all components: for every $i$, we have $\|\mathbf{v}_i\| = \|\mathbf{v}_h\|$  and $\|\mathbf{T}_i \mathbf{v}_i\| = \|\mathbf{v}_i\|$. But this leads to a contradiction. Specifically, by assumption, $\rho(\mathbf{T}_{\overline{k}}) < 1$. Therefore, for any nonzero vector $\mathbf{v}_{\overline{k}}$, it must hold that  $\|\mathbf{T}_{\overline{k}} \mathbf{v}_{\overline{k}}\| < \|\mathbf{v}_{\overline{k}}\|$,
contradicting the earlier conclusion that $\|\mathbf{T}_{\overline{k}} \mathbf{v}_{\overline{k}}\| = \|\mathbf{v}_{\overline{k}}\|$. Hence, our initial assumption that $\rho(\mathbf{B}) \geq 1$ must be false. We then conclude that $\rho(\mathbf{B}) < 1$, and the recursion \eqref{eq:mean_recursion_dist} converges to zero by  contraction arguments. This concludes the proof of Theorem 2.

\end{document}